\documentclass[10pt,a4paper,reqno]{amsart} 
\usepackage{amsaddr}
\usepackage{amsmath}
\usepackage{bbm}
\usepackage{mathtools}
\usepackage{amssymb}
\usepackage{eurosym}
\usepackage{graphicx}
\usepackage{color}
\usepackage{latexsym}
\usepackage[utf8]{inputenc}
\usepackage[T1]{fontenc}
\usepackage{comment} 
\usepackage{stmaryrd}
\usepackage{enumitem}
\usepackage{amssymb} 
\usepackage[mathcal]{eucal}
\usepackage{fancybox}
\usepackage{xcolor}
\usepackage{hyperref}
\hypersetup{
       colorlinks   = true,
       citecolor    = green!50!black
     }

\usepackage{palatino}

\newcommand{\fig}[3]{
\begin{figure}[h!]
\begin{center}
 \includegraphics #1
 \end{center}
\vspace{-7pt}
\caption{ #2}
\label{#3}
\end{figure}
}

\newtheorem{proposition}{Proposition}[section]

\newtheorem{ex}[proposition]{Example}
\newtheorem{lem}[proposition]{Lemma}
\newtheorem{defi}[proposition]{Definition}

\newtheorem{theo}[proposition]{Theorem}

\newtheorem{cor}[proposition]{Corollary}

\usepackage{cleveref}

\usepackage{comment}

\usepackage{mathtools}
\def\multiset#1#2{\ensuremath{\left(\kern-.3em\left(\genfrac{}{}{0pt}{}{#1}{#2}\right)\kern-.3em\right)}}

\catcode`\@=11
\def\section{\@startsection{section}{1}%
 \z@{.7\linespacing\@plus\linespacing}{.5\linespacing}%
 {\normalfont\bfseries\scshape\centering}}

\def\subsection{\@startsection{subsection}{2}%
  \z@{.5\linespacing\@plus\linespacing}{.5\linespacing}%
  {\normalfont\bfseries\scshape}}

\def\subsubsection{\@startsection{subsubsection}{3}%
 \z@{.5\linespacing\@plus\linespacing}{-.5em}%{.5\linespacing}%
  {\normalfont\bfseries\itshape}}
\catcode`\@=12

\graphicspath{{./images/}}

\begin{document}
\title{Next-to${}^k$ leading log expansions by chord diagrams} 
\author{Julien Courtiel}
\address{ \vspace*{-.5cm} \small Normandie University, UNICAEN, ENSICAEN, CNRS, GREYC, \\{\tt julien.courtiel@unicaen.fr} }
\thanks{JC is supported by the French INS2I JCJC Grant \textit{ASTEC}.  KY is supported by an NSERC Discovery grant and the Canada Research Chair program.  Thanks to the referee for their insightful comments.}
\author{ \vspace*{-.5cm}  Karen Yeats }
\address{ \vspace*{-.5cm}  \small University of Waterloo, \\{\tt kayeats@uwaterloo.ca }}

\begin{abstract}
  Green functions in a quantum field theory can be expanded as bivariate series in the coupling and a scale parameter.  The leading logs are given by the main diagonal of this expansion, i.e. the subseries where the coupling and the scale parameter appear to the same power; then the next-to leading logs are listed by the next diagonal of the expansion, where the power of the coupling is incremented by one, and so on.

  We give a general method for deriving explicit formulas and asymptotic estimates for any next-to${}^k$ leading-log expansion for a large class of single scale Green functions.  These Green functions are solutions to Dyson-Schwinger equations that are known by previous work to be expressible in terms of chord diagrams.  We look in detail at the Green function for the fermion propagator in massless Yukawa theory as one example, and the Green function of the photon propagator in quantum electrodynamics as a second example, as well as giving general theorems.  Our methods are combinatorial, but the consequences are physical, giving information on which terms dominate and on the dichotomy between gauge theories and other quantum field theories.
\end{abstract}

\maketitle

\section{Introduction}

To calculate physical amplitudes in perturbative quantum field theory it suffices to understand the renormalized one-particle irreducible (1PI) Green functions in the theory.  As series one can obtain the 1PI Green functions by applying renormalized Feynman rules to sums of Feynman graphs. Choosing an external scale parameter $L$ we can think of the Green functions as multivariate series $G(x, L, \theta)$ where $x$ is the perturbative expansion parameter and $\theta$ are dimensionless parameters capturing the remaining kinematic dependence.  Our examples will all be Green functions for propagator corrections, so we can take $L=\log q^2/\mu^2$ where $q$ is the momentum flowing through, $\mu$ is the reference scale for renormalization.  Furthermore, take $x$ to be the coupling constant, and do without any $\theta$.

These Green functions will always be triangular in the sense that the power of $L$ is at most the power of $x$ in any term.  It is often useful to consider first the leading logs in $G(x,L)$, that is the subseries where $x$ and $L$ come with the same power, then the next-to leading logs which are suppressed by one power of $x$ and so on.  As a whole this gives the log expansion as in the following definition.

\begin{defi}
  Write
\[
G(x,L) = 1 + \sum_{k \geq 1} H_k(xL)x^k = 1 + \sum_{k\geq 1}\sum_{i\geq 0} g_{i, i+k}(xL)^i x^k 
\]
Then
\[
H_k(z) = \sum_{i\geq 0} g_{i, i+k}z^i
\]
is the next-to${}^k$ leading log (N${}^k$LL) part of $G(x,L)$.
\end{defi}

Our results give an automatable and combinatorial method to obtain closed forms for any particular $H_k$ coming from a broad family of propagator-type Dyson-Schwinger equations.  The general shape of solutions only depends on a single parameter of the Dyson-Schwinger equation which controls how the number of insertion places grows as the loop order grows.  Different theories with isomorphic diagrammatics only show their difference in their dependence on the coefficients of the expansions of the regularized Feynman integrals of the primitive\footnote{Here primitive means primitive in the renormalization Hopf algebra, or equivalently having no proper subdivergences.} Feynman graphs of the theory contributing to the Dyson-Schwinger equation.  These coefficients are the $f_i$ of \cite{MYchord} and our previous work \cite{CYchord} and are the $a_{k,i}$ of \cite{HYchord}, a notation which we will also use in the present paper.
The way these coefficients come in is well-controlled; for example in all cases the full $H_k$ depends only on coefficients coming from the first $k$ terms of expansions of primitive graphs with at most $k+1$ loops (this is clear and long-known in physics and is also an immediate corollary of the main results of \cite{MYchord} and \cite{HYchord}).  Furthermore, we can determine for every $k$ which terms will dominate $H_k$.

Kr\"uger and Kreimer in \cite{KKllog} also give an analysis of log expansions of Green functions from Dyson-Schwinger equations with similar goals and outcomes.  The underlying combinatorial structures we use are quite different, with the consequence that some aspects are clearer to see in our perspective.  In particular, our approach is more fully automatable and better suited to resolve questions of domination or asymptotic analysis.  A fuller comparison of our work with theirs is given in Section~\ref{subsec compare}.

Our underlying combinatorial structures are decorated rooted connected chord diagrams.  Previous work of one of us with Marie \cite{MYchord} and Hihn \cite{HYchord} showed how to solve Dyson-Schwinger equations as sums indexed by these chord diagrams.  To see that the chord diagram expansions are physically interesting, compare them to the original Feynman diagram expansions of the Green functions.  The latter are also sums indexed by combinatorial objects, but each object contributes in a highly nontrivial manner, namely it contributes the renormalized Feynman integral.  The chord diagram expansions are indexed by decorated rooted connected chord diagrams, but each diagram contributes a small polynomial in $L$ with coefficients monomials in the numbers from the expansions of the Feynman integrals of the primitives.
% with purely combinatorial coefficients. 
 So the only analytic content is in the Feynman integrals of the primitives; the rest simply consists of putting the coefficients of the expansions of the primitives together in purely combinatorial ways.  Furthermore, the shape of the answer depends only on this insertion growth parameter mentioned above, so there is a general form for different theories with the same insertion properties in their diagrammatics.

The chord diagram expansions are interesting for physics because they separate the analytic and combinatorial information.  The N${}^k$LL log expansions are a particularly good example of this as the chord diagram expansion lets us read off the various log expansions as sums of controllable monomials, rather than hiding pieces among different Feynman diagrams.  We give our work in this paper as a specific example of how the chord diagram expansions can give concrete, physically interesting results.

Note, unusually, that these chord diagrams do \emph{not} come from any sort of Wick pairing.  The rigorous proofs of these chord diagram expansions \cite{MYchord, HYchord} are inductive and do not explain why these objects appear.

Previous work of both of us \cite{CYchord} considered the original chord diagram expansion of \cite{MYchord} which applies only in the case of one particular shape of Dyson-Schwinger equation.  We were able to prove statistical properties of various chord diagram parameters which are relevant to the expansion and we were able to understand the log expansions.  This called for an extension to more general propagator-type Dyson-Schwinger equations; these have chord diagram expansions from \cite{HYchord}.  Herein we study the log expansions in all these cases; thus obtaining an automatable, combinatorial method for computing the $H_k$ for physically relevant propagator Dyson-Schwinger equations.

\section{Background}

\subsection{Dyson-Schwinger equations}

We are interested in Dyson-Schwinger equations for 1PI Green functions for massless propagator corrections in a quantum field theory.  At the present time we do not have analogous tools available for vertex corrections or for coupled systems of Dyson-Schwinger equations, nor do we consider masses.  This is discussed further below.

Thus, the Green functions we deal with are all single scale, and so we will view them as functions of two variables, the coupling constant $x$ and the kinematical scale variable $L=\log(q^2/\mu^2)$ where $q$ is the external momentum and $\mu$ is the reference scale with respect to which we renormalize by subtraction\footnote{Note that because of an unfortunate sign convention originating in \cite{kythesis} and carried throughout related work, the sign of $L$ is opposite to what would be in most of the physics literature, and so it is perhaps best to view $L$ as $-\log(q^2/\mu^2)$.  In particular this leads to a sign difference in the comparison with the work of Kr\"uger and Kreimer in Section~\ref{subsec compare}.}.
The corresponding Dyson-Schwinger equations take the form of integral functional equations for the Green function where the integral kernels are the Feynman integrals of the primitive skeleton diagrams into which we insert and the recursive appearances of the Green function correspond to these insertions.  The recursive appearances of the Green function will always be in the denominator since these are propagator insertions and so we are working with the inverse propagator, or equivalently we are inserting geometric series in the propagator correction diagrams.

However, this is not the form of the Dyson-Schwinger equation that was most convenient for \cite{MYchord, HYchord} where the chord diagram expansion is developed.  We proceed to transform the equation first by expanding out the recursive appearances of the Green function as formal series in two variables, expanding both the Green functions themselves and the geometric series in them.  This gives logarithms of the momenta of the edges on which we insert in the numerator, and we trade these for powers and derivatives using $\partial^\ell_\rho (p^2)^\rho|_{\rho=0} = \log^\ell(p^2)$.  Then exchanging the derivatives with the integral and reassembling the series, we have brought the recursive appearances of the Green function outside the integral at the cost of replacing their kinematical argument by a differential operator.  The equation we are left with has the form
\[
G(x, L) = 1 - \sum_{k \geq 1}x^kG(x, \partial_{-\rho})^{1-sk}(e^{-L\rho}-1)F_k(\rho)|_{\rho=0}
\]
where $G(x,L)$ is the Green function, $F_k(\rho)$ is the sum of the Feynman integrals of the $k$-loop primitives regularized by $\rho$ on the edge or edges where we insert, and $s$ is a parameter counting how the number of insertion places grows.  An example of these manipulations can be found in \cite{kythesis} or \cite{Ymem} as a running example (see in particular Example 3.5).

\medskip

The $s$ parameter requires some additional explanation.  It captures how the number of insertion places grows with the loop order.  For example, the one loop photon correction diagram in QED has no internal photon edges so no insertion places for itself.  The two loop photon correction diagram in QED has one internal photon edge, so one insertion place.  Any three loop photon correction diagram has two internal photon edges, and so on.  Here the number of insertion places grows by 1 as the loop order increases by 1, and so $s=1$.  In contrast the one loop fermion correction diagram either in QED or in Yukawa theory has one internal fermion edge, while any two loop fermion correction diagram in such a theory has three internal fermion edges, and any three loop fermion correction has five.  Thus in this case $s=2$.  Similar counting for the propagator in $\phi^3$ theory gives $s=3$.

It is perhaps surprising that the $s$ parameter is important, but it is because it controls the structure and growth of the Feynman diagrams that are built by the Dyson-Schwinger equation.  This importance was already seen in \cite{HYchord}, but here looking at the asymptotics we see an even more striking manifestation of the importance of $s$, in that we find a dichotomy between $s=1$ and $s>1$.  This dichotomy appears throughout our results in subsequent sections and some thoughts on the physical meaning of it are discussed in the final section.

\medskip

Additionally, observe that the Feynman integrals for the primitives are regularized by a single parameter $\rho$.  This means that a single parameter is involved in all propagators into which the diagrammatic version of the Dyson-Schwinger equation inserts.  When there are multiple insertion places, we take each possibility for the Feynman integral for the primitive regularized at one of the insertion places and then sum these and divide by the number of insertion places. For more on symmetric insertion see Section 2.3.3 of \cite{kythesis} or \cite{Ymem}.  Symmetric insertion gives the sum over all Feynman diagrams built by inserting into each of the insertion places used in the symmetric insertion, but we are constrained to get them all together as a linear combination.

\medskip

Taking all of the above into consideration, the Dyson-Schwinger equations that were able to be solved with chord diagrams in \cite{HYchord} remain restricted to the single scale case, and so correspond only to Dyson-Schwinger equations for propagator insertions.  Within the context of propagator insertions, they are general whenever symmetric insertion suffices.  Current work of one of us with a student, Lukas Nabergall, looks to avoid the need for symmetric insertion; we expect similar chord diagram descriptions of the solution to hold there as well, and as a consequence to prove that symmetric insertion was in retrospect no restriction.

Additionally, the work of \cite{HYchord}, and hence our work here, does not deal with coupled systems of Dyson-Schwinger equations, nor with nonzero masses.  We expect that similar results will hold for systems, but this remains work for the future.  Vertex insertions also remain for the future and will be more complicated because of the various angles involved, though we hope in the spirit of \cite{BrK}, that it will also prove manageable.

Nonetheless, the currently available chord diagram expansions are physically relevant, capturing Dyson-Schwinger equations for propagator corrections, and consequently understanding more about the log expansions and their asymptotics in these cases is of interest.

\subsection{Chord diagrams and chord diagram expansions}

In \cite{HYchord}, Hihn and one of us proved that the solution of the Dyson-Schwinger equation
\begin{equation*}
G(x, L) = 1 - \sum_{k \geq 1}x^kG(x, \partial_{-\rho})^{1-sk}(e^{-L\rho}-1)F_k(\rho)_{\rho=0}
\end{equation*}
can be seen as a sum indexed by some combinatorial objects: the \textit{decorated connected chord diagrams}. The description of this function is quite intricate and uses the coefficients $a_{k, i}$ of the Laurent series expansion of $F_k(\rho) = \sum_{i\geq 0} a_{k, i} \rho^{i-1}$. The full description will be displayed at Theorem~\ref{theo:startingpoint}, but many definitions are required beforehand. Note that the statement of this theorem has been updated since \cite{HYchord}: some parameters intervening in the chord diagram expansions (notably the weight of a diagram) have been clarified in \cite{coyeze}.

\begin{defi}[Chord diagram] A \emph{chord diagram} (or \emph{diagram} in short) is a collection of ordered pairs $(a_1,b_1), \dots, (a_n,b_n)$ such that  $\{a_1,b_1,\ldots,a_n,b_n\} = \{1,2,\ldots, 2n\}$  and such that for each $1 \le i \le n$, we have $a_i < b_i$. The pairs $(a_i,b_i)$ are the \emph{chords} of the diagram. The \emph{root chord} is the unique chord whose first component is $1$. The number of chords in $C$ is denoted $|C|$.
\end{defi}

% In other words, a chord diagram of size $n$ is a perfect matching of the set $\{1,2,\ldots, 2n\}$. 

Chord diagrams are often represented by drawing dots and arches: starting from a horizontal line of $2|C|$ dots, we draw for each chord $(a_i,b_i)$ an arch that links the $a_i$th leftmost dot to the $b_i$th one. This is referred as the \textit{linear representation} of a chord diagram (see \Cref{ex_diagram} for examples). Another convention exists: the \textit{circular representation}. In this convention, the $2|C|$ dots, instead of being aligned, are drawn on an oriented circle. The circular representation notably appears in~\cite{MYchord,HYchord} and explains the chord language, but we are going to adopt the linear convention for this paper, as we did in \cite{CYchord,coyeze}.

\fig{[scale=2]{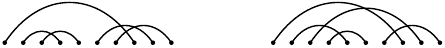}}{\textit{Left.} Linear representation of the chord diagram $(1,8),(2,4),(3,5),(6,9),(7,10)$. \textit{Right.} Linear representation of the diagram $(1,8),(2,5),(3,9),(4,6),(7,10)$.}{ex_diagram}

\begin{defi}[Connected chord diagram]  The \emph{intersection graph} of a chord diagram $D$ is the directed graph constructed as follows: the vertices are the chords of $D$; and we have a directed edge from a chord $(a,b)$ to another chord $(a',b')$ if $a < a' < b < b'$, in which case the two chords are said to \emph{intersect}. A chord diagram is \emph{connected} if its intersection graph is (weakly) connected. The \emph{connected components} of a diagram are the sets of chords corresponding to the (weakly) connected components of the intersection graph.
\end{defi}

Note that the condition $a < a' < b < b'$ implies that chords $(a,b)$ and $(a',b')$ are crossing in the linear representation of the diagram. Consequently, a diagram is connected when the drawing of its linear representation is ``in one piece''.  
For example, the left diagram of \Cref{ex_diagram} is not connected since there are two connected components ($(1,8),(6,9),(7,10)$ on the one hand, $(2,4),(3,5)$ on the other hand), while the right diagram is connected. 

Note also that given two crossing chords, the definition above says that the directed edge between those chords goes from the chord with the smaller left end point to the chord with the larger left end point.

\begin{defi}[Terminal chord] A chord $(a,b)$ in a diagram is \emph{terminal} if every chord $(c,d)$ intersecting $(a,b)$ satisfies $c < a$. 
\label{def:terminal}
\end{defi}
Equivalently, a chord $(a,b)$ is terminal if there is no chord $(a',b')$ satisfying $a < a' < b < b'$. It means that terminals chords correspond to the vertices with no outgoing edge in the intersection graph of their diagram. 
For example, chords $(4,6)$ and $(7,10)$ are terminal in the right diagram of \Cref{ex_diagram}, and chords $4,6,7$ are terminal in the diagram of \Cref{covering}.

\begin{defi}[Intersection order] The \emph{intersection order} of the chords of a rooted connected diagram $C$ is defined as follows.
\begin{itemize}
  \item The root chord of $C$ is the first chord in the intersection order.
  \item Remove the root chord of $C$ and let $C_1, C_2, \ldots, C_k$ be the connected components of the result ordered by their leftmost vertex.
  \item For the intersection order of $C$, after the root chord come all the chords of $C_1$ ordered recursively in the intersection order, then all the chords of $C_2$ ordered by intersection order, and so on.
\end{itemize}
\label{def:intersection}
\end{defi}

\fig{[scale=1]{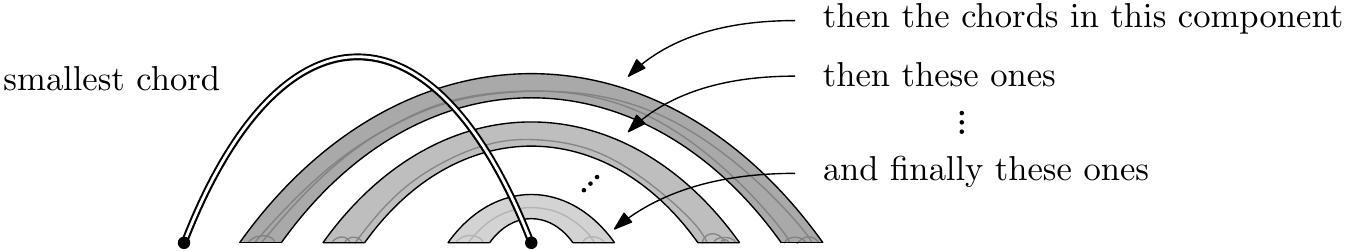}}{Illustration of the intersection order of a connected diagram.}{intersection}

The notion of intersection order is schematized  by \Cref{intersection}. For instance, the intersection order for the right diagram of \Cref{ex_diagram} is  $(1,8) < (2,5) < (3,9) < (7,10) < (4,6)$. Remark that the intersection order does \textit{not} correspond to the left-right ordering since $(4,6)$ is after $(7,10)$ for the intersection order in the previous diagram.

One can also remark that the intersection order is compatible with the intersection graph.  To this end, recall that if an edge goes from chord $c$ to chord $c'$ in the intersection graph, then the first endpoint of $c$ is smaller than the first endpoint of $c'$. However, the direction could just as well be defined in terms of the intersection order: if an edge goes from a chord $c$ to another chord $c'$ in the intersection graph, then we have $c < c'$ for the intersection order. Consequently, the last chord for the intersection order is always terminal.

From now on, we identify the chords with their positions in the intersection order (as one can see at \Cref{covering}). Thus, chord $1$ always denotes the root chord of a diagram.

\begin{defi}[Covering of a chord] \label{def:covering}
  Let $C$ be a connected diagram viewed through its linear representation. We call \emph{interval} the space between two consecutive dots. 
  
We compute the \emph{covering} of the chords of the diagram $C$ as follows. Begin by $i=1$. 
Label by $i$ all the intervals below chord $i$, replacing any previous labels. Then increase $i$ by $1$, and repeat the previous procedure until every chord has been processed in this way.

At the end of this procedure, the intervals are partitioned among the chords according to their labels.   For $i \in \{1,\dots,|C|\}$, the intervals labeled by $i$ are said to be \emph{covered by chord $i$}. Let $\omega(i)$ be the number of such intervals, minus $1$. We call this number the \emph{covering number} of chord $i$.
\end{defi}

An example of this construction is shown by \Cref{covering}. For this diagram, we have $\omega(1) = \omega(2) = 0$, $\omega(3)=\omega(5)=\omega(6)=\omega(7) = 1$ and $\omega(4) = 2$.

\fig{[scale=1.7]{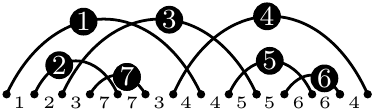}}{Connected diagram, indices of the chords for the intersection order, and covering of the chords}{covering}

The covering numbers will be used to define a ``weight'' associated to each diagram (cf. De\-fi\-nition~\ref{def:weight}).

\begin{defi}[Decorated diagram]
A chord diagram $C$ is \emph{decorated} if $C$ is equipped with a mapping $d : \{\textrm{chords of }C\} \rightarrow \mathbb{N}_{>0}$. Integer $d(c)$ is called the \emph{decoration} of chord $c$. The \emph{size} of a decorated chord diagram is the sum of the decorations and is denoted $\|C\|$.
\label{def:decorated}
\end{defi}

Do not mistake $\|C\|$ for $|C|$, which is the number of chords. An example of decorated diagram is depicted in \Cref{decorated}. To facilitate the reading, a chord with decoration $2$ has been represented as a double arch, a chord with decoration $3$ as a triple arch, while a chord with decoration $1$ has been unchanged. Thus, for this diagram, we have $d(1) = d(2) = d(5) = d(6) = 1$, $d(4) = d(7) = 2$ and $d(3) = 3$.

\fig{[scale=1.7]{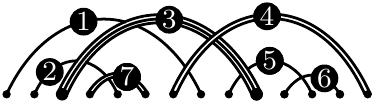}}{Decorated diagram of size $11$}{decorated}

\begin{defi}[Weight of a decorated diagram]
Let $C$ be a decorated connected diagram, and $s$ a positive integer. The \textit{weight} of $C$, denoted by $w(C)$, is the integer
\begin{equation}
w(C) = \prod_{c=1}^{|C|}\binom{d(c)s + \omega(c) -2}{\omega(c)},
\label{eq:defw}
\end{equation}
where $d(c)$ is the decoration of chord $c$, and $\omega(c)$ its covering number  (cf. Definition~\ref{def:covering}). 
\label{def:weight}
\end{defi}
Note that the weight depends implicitly on $s$, the insertion parameter from the Dyson-Schwinger equation discussed in more detail in the previous section. For example, if we consider $s=2$, the weight of the diagram $C$ in \Cref{decorated} is
\[ \prod_{c=1}^{|C|}\binom{2d(c) + \omega(c) -2}{\omega(c)} = \binom{0}{0}^2 
\binom{5}{1} \binom{4}{2} \binom{1}{1}^2 \binom{3}{1} = 90.  \]

We can finally state the theorem that is at the root of this article.

\begin{theo}[\cite{HYchord,coyeze}] 
\label{theo:startingpoint}
Let $a_{k, i}$ be the coefficients of the Laurent expansion of $F_k(\rho) = \sum_{i\geq 0} a_{k, i} \rho^{i-1}$. The Dyson-Schwinger equation 
\begin{equation}
G(x, L) = 1 - \sum_{k \geq 1}x^kG(x, \partial_{-\rho})^{1-sk}(e^{-L\rho}-1)F_k(\rho)_{\rho=0}
\label{eq:gen case}
\end{equation}
has for solution the series
\begin{equation}
G(x,L) = 1 - \sum_{\substack{\textrm{decorated connected}\\\textrm{diagram }C}} \left(\sum_{i = 1}^{t_1(C)} a_{d(t_1(C)), t_1(C)-i} \frac{(-L)^i}{i!} \right) w(C)
A(C) x^{\|C\|}.
\label{eq:solutionDS}
\end{equation}
%In this case we must sum 
where $t_1(C)<t_2(C)<\dots<t_\ell(C) = |C|$ lists the positions of all the terminal chords (cf. De\-fi\-nition~\ref{def:terminal}) in intersection order (cf. De\-fi\-nition~\ref{def:intersection}); $d(c)$ is the decoration of chord $c$ (cf. De\-fi\-nition~\ref{def:decorated}); $w(C)$ denotes the weight of diagram $C$ (cf. De\-fi\-nition~\ref{def:weight});
and  
\begin{equation} A(C) = \prod_{\substack{c \text{ not}\\ \text{terminal}}}a_{d(c), 0}\prod_{j=2}^{\ell}a_{d(t_j(C)), t_j(C)-t_{j-1}(C)}
\label{eq:defA}
\end{equation}
($\ell$ is the number of terminal chords).
\end{theo}

\begin{ex}
Equation \eqref{eq:solutionDS} can be seen as a generating function of decorated connected chord diagrams weighted by some product of $a_{k,i}$ and binomials. To clarify this a bit, let us compute the contribution of the diagram $C$ in \Cref{decorated} and $s=2$. This diagram has three terminal chords with $t_1(C) = 4$, $t_2(C)=6$ and $t_3(C) = 7$, hence we have
\[A(C) = a_{1,0}^3\,a_{3,0}\,a_{1,2}\,a_{2,1}.\]
Moreover, we saw that the weight of the diagram is $90$.
Thus the contribution of $C$ in $G(x,L)$ is:
\[-90\,a_{1,0}^3\,a_{3,0}\,a_{2,1}\,a_{1,2}\,\left( a_{2,3} L + a_{2,2} \frac{L^2}{2} +  a_{2,1} \frac{L^3}{6} + a_{2,0} \frac{L^4}{24} \right)x^7. \]
\end{ex}

\begin{cor} For every $k \geq 0$, the next-to$^k$ leading-log expansion is equal to
\[ H_k(z) = - \hspace*{-15pt} \sum_{\substack{\textrm{decorated connected}\\\textrm{diagram }C\textrm{ such that}\\ \|C\| - t_1(C) \leq k \textrm{ and }\|C\| > k }} \hspace*{-15pt} a_{d(t_1(C)), t_1(C)+k-\|C\|}  w(C)
A(C) \frac{(-z)^{\|C\|-k}}{{(\|C\|-k)!}},\]
where the parameters are as defined in \Cref{theo:startingpoint}. 
\label{cor:hk}
\end{cor}

\begin{proof}To find $H_k$, we have to extract the coefficient in $ (Lx)^{\|C\|-k} x^k$ in \Cref{eq:solutionDS}. For each diagram $C$, the corresponding
 coefficient is non null if and only if $i = \|C\|-k \in \{1,\dots,t_1(C)\}$.
This is equivalent to the condition $\|C\| - t_1(C) \leq k$ and $\|C\| > k $.
Finally we substituted $i$ by $\|C\|-k$, which gives the above formula.
\end{proof}

\subsection{Description in terms of weighted generating functions}

Corollary~\ref{cor:hk} gives an explicit formula for the next-to$^k$ leading-log expansion, but it is still not very practicable for the moment. We are going to rewrite it in terms of weighted generating functions of particular chord diagrams.

\begin{defi}[Weighted generating function]
Given $\mathcal A$ a family of decorated chord diagrams, the \textit{weighted generating function} of $\mathcal A$ is defined as
\[\sum_{C \in \mathcal A} w(C) \frac{z^{\|C\|}}{\|C\|!} \]
where $w(C)$ is defined in Definition~\ref{def:weight}.
\end{defi}

\begin{ex} Assume that $s=2$ and $\mathcal A$ is the set of decorated connected chord diagrams of size $\leq 3$ with a chord of decoration $2$. Then the cardinality of $A$ is $3$: we have the one-chord diagram (decoration $2$), weighted by $1$; then the only two-chords connected diagram where the first chord is weighted by $2$, and the other chord by $1$ -- the weight of this diagram is also $1$; and finally the same diagram as before but where the decorations between the two chords have been swapped  -- this one is weighted by $3$. The weighted generating function of $\mathcal A$ is then $\dfrac{z^2}{2!} + \dfrac{z^3}{3!} + 3\, \dfrac{z^3}{3!} = \dfrac 1 2 z^2 + \dfrac 2 3 z^3$.
\end{ex}

We wish now to partition the diagrams into families such that $a_{d(t_1(C)), t_1(C)+k-\|C\|} A(C)$, which is present in Corollary~\ref{cor:hk}, is almost constant in every family. This is the purpose of the following definition.

\begin{defi}[Type of a decorated connected chord diagram] Let $C$ be a decorated connected diagram. The \emph{type} of $C$ is the triplet $\left( \Delta, \mathcal G, \mathcal D \right)$, where:
\begin{itemize}
\item $\Delta$ is the decoration of the first terminal chord, also denoted by $d(t_1(C))$;
\item $\mathcal G$ is the multiset $\{\left(d(t_j(C)),t_j(C)-t_{j-1}(C)\right)\}$ for $1 < j \leq \ell$  (remember that $\ell$ is the number of terminal chords, $d(t_j(C))$ is the decoration of the $j$th terminal chord and $t_j(C)-t_{j-1}(C)$ is the gap between the positions of the $j$th and $(j-1)$th terminal chord in the intersection order);
\item $\mathcal D$ is the multiset formed by the decorations of the chords that are neither terminal nor of decoration $1$.
\end{itemize}

\end{defi}

The point of this definition is that any number of non-terminal chords decorated by $1$ can be added in while not affecting the type.

For example, the type of the diagram from \Cref{decorated} is $(2,\{(1,2),(2,1)\},\{3\})$.

\begin{lem} 
 Let $C$ be a diagram with type $(\Delta,\mathcal G,\mathcal D)$. Then the position of the first terminal chord of $C$ is
\begin{equation}
\label{eq:t1}
t_1(C) = \|C\| - \sum_{(d,g) \in \mathcal G} (d + g - 1) - \sum_{d \in \mathcal D} (d - 1) - \Delta + 1 ;
\end{equation}
where the sums are considered with multiplicity. Therefore, over all diagrams of type  $(\Delta,\mathcal G,\mathcal D)$, the monomial $a_{d(t_1(C)), t_1(C)+k-\|C\|} A(C)$ is equal to
\[ {a_{1,0}}^{|C|-1} a_{\Delta, k - \sum_{(d,g) \in \mathcal D} (d + g - 1) - \sum_{d \in \mathcal G} (d - 1) - \Delta + 1}  \prod_{d \in \mathcal D} \frac{a_{d, 0}}{a_{1,0}} \prod_{ (d,g) \in \mathcal G} \frac {a_{d,g}}{a_{1,0}}\]
and only depends on $k$, the type and the number of chords of $C$.
\label{lem:t1}
\end{lem}

Note that the statement makes sense even when $a_{1,0}=0$ as the total power of $a_{1,0}$ in the expression is always nonnegative.

\begin{proof} Let us consider $C$ a diagram as stated by the lemma.
Using the definition of $\mathcal G$ and the fact that $t_\ell(C) = |C|$ (the last chord in the intersection order is always terminal), we deduce
\[\sum_{(d,g) \in \mathcal G} g = \sum_{j=2}^{\ell} (t_{j}(C) - t_{j-1}(C)) = t_\ell(C) - t_1(C) = |C| - t_1(C).\]
 Moreover, 
\[\|C\| = \sum_{c=1}^{|C|} d(c) = |C| + \sum_{\textrm{chord }c} ( d(c) - 1 ) = |C| + \sum_{(d,g) \in \mathcal G} (d - 1) + \sum_{d \in \mathcal D} (d - 1) + \Delta - 1 \]
(The terms indexed by chords $i$ such that $d(i)=1$ are vanishing.)
Eliminating $|C|$ from the two above equations, we prove \Cref{eq:t1}. The expression for $a_{d(t_1(C)), t_1(C)+k-\|C\|} A(C)$ is then straightforward.
\end{proof}

\begin{theo}Let $k \geq 0$. For any triplet $(\Delta,\mathcal G,\mathcal D)$ such that
\begin{itemize}
\item $\Delta$ is a positive integer;
\item $\mathcal G$ is a set of pairs $(d,g)$, where $d > 0$ and $g > 0$
\item $\mathcal D$ is a set of integers greater than $1$,
\end{itemize}
we define $t(\Delta,\mathcal G,\mathcal D)$ as the integer $ \sum_{(d,g) \in \mathcal G} (d + g - 1) + \sum_{d \in \mathcal D} (d - 1) + \Delta - 1$. Then the next-to${}^k$ leading-log expansion
\[ H_k(z) = \kappa - \hspace*{-15pt} \sum_{\substack{\textrm{all possible }(\Delta,\mathcal G,\mathcal D)\\\textrm{ such that } t(\Delta,\mathcal G,\mathcal D) \leq k }}  \left( {a_{1,0}}^{k-\Delta}  a_{\Delta, k - t(\Delta,\mathcal G,\mathcal D)} \prod_{d \in \mathcal D} \frac{a_{d, 0}}{a_{1,0}^d} \prod_{ (d,g) \in \mathcal G} \frac {a_{d,g}}{a_{1,0}^d}
\right) {\dfrac {\partial^k F_{\Delta,\mathcal G,\mathcal D}}{\partial z^k}}\left( -a_{1,0} z \right),\]
where $F_{\Delta,\mathcal G,\mathcal D}$ is the weighted generating function of diagrams of type $(\Delta,\mathcal G,\mathcal D)$, and $\kappa$ is the constant which cancels out the right-hand side when it is evaluated at $z=0$.
\label{theo:intermsofwgf}
\end{theo}

Remark that each summand in $t(\Delta,\mathcal G,\mathcal D)$ is positive, so the sum indexed by types such that $t(\Delta,\mathcal G,\mathcal D) \leq k$ is always finite. We have thus expressed the next-to${}^k$ leading-log expansion in terms of a finite sum of weighted generating functions (that we will be able to compute).

\begin{proof} In the sum of the formula of Corollary~\ref{cor:hk}, we gather diagrams according to their type. The types are triplets $(\Delta,\mathcal G,\mathcal D)$ as defined in the theorem such that $t(\Delta,\mathcal G,\mathcal D) \leq k$. Using Lemma~\ref{lem:t1}, we further factorize since $a_{d(t_1(C)), t_1(C)+k-\|C\|}  A(C)$ can be expressed in terms of the type. We obtain at this point the formula
\[H_k(z) =  - \hspace*{-25pt} \sum_{\substack{\textrm{all possible }(\Delta,\mathcal G,\mathcal D)\\\textrm{such that } t(\Delta,\mathcal G,\mathcal D) \leq k }}  \hspace*{-15pt} a_{\Delta, k - t(\Delta,\mathcal G,\mathcal D)} \prod_{d \in \mathcal D} \frac{a_{d, 0}}{a_{1,0}} \prod_{ (d,g) \in \mathcal G} \frac {a_{d,g}}{a_{1,0}} 
\sum_{\substack{\textrm{decorated connected}\\\textrm{diagram } C\textrm{ with type}\\(\Delta,\mathcal G,\mathcal D)\textrm{ s.t. } \|C\| > k }} \hspace*{-15pt} {a_{1,0}}^{|C|-1} w(C) \frac{(-z)^{\|C\|-k}}{{(\|C\|-k)!}}.\]
With the substitution $y = - a_{1,0} z$, the sum at the far right can be 
rewritten as 
\begin{equation}
\sum_{\substack{\textrm{decorated connected}\\\textrm{diagram } C\textrm{ with type}\\(\Delta,\mathcal G,\mathcal D)\textrm{ s.t. } \|C\| > k }} a_{1,0}^{|C| + k - \|C\| - 1} w(C) \frac{y^{\|C\|-k}}{(\|C\|-k)!}.
\label{eq:subsum}
\end{equation}
Note that $|C| + k - \|C\| - 1$ is a constant over all diagrams $C$ with type $(\Delta,\mathcal G,\mathcal D)$. Indeed, using the fact that the difference between $\|C\|$ and $|C|$ is computed by summing the decorations minus 1 over all chords of $C$, this number is always equal to $k - \Delta - \sum_{d \in \mathcal D} (d-1) - \sum_{(d,g) \in \mathcal G} (d-1)$.

Moreover, up to a constant (which is a power of $a_{1,0}$), we recognize in \eqref{eq:subsum} the $k$th derivative of  $F_{\Delta,\mathcal G,\mathcal D}$, except that the constant term of \eqref{eq:subsum} is always $0$. 
In order to fix the discrepancy of the constant term, we add a constant $\kappa$ which can be computed afterwards. One then recovers the formula stated by the theorem.
\end{proof}

\begin{ex} Let us consider $k=0$. There is only one triplet $(\Delta,\mathcal G,\mathcal D)$ such that $t(\Delta,\mathcal G,\mathcal D)=0$ (every summand must be null): it is $(1,\varnothing,\varnothing)$. So
\[H_0(z) = -F_{1,\varnothing,\varnothing}(-a_{1,0}\,z)\]
For $k=1$, there are four possible triplets: $(1,\varnothing,\varnothing)$, $(2,\varnothing,\varnothing)$, $(1,\{(1,1)\},\varnothing)$ and $(1,\varnothing,\{2\})$.
Therefore
\begin{multline*}H_1(z) =  a_{1,1} - a_{1,1} \dfrac {\partial F_{1,\varnothing,\varnothing}}{\partial z}(-a_{1,0}\,z) -  \frac {a_{2,0}} {a_{1,0}}   \dfrac {\partial F_{2,\varnothing,\varnothing}} {\partial z}    (-a_{1,0}\,z)\\ - \frac {a_{1,1}} {a_{1,0}}  \dfrac {\partial F_{1,  \{(1,1)\} ,\varnothing}}  {\partial z}   (-a_{1,0}\,z) - \frac {a_{2,0}} {a_{1,0}^2}  \dfrac {\partial F_{1,\varnothing,\{2\}}}   {\partial z}(-a_{1,0}\,z) .
\end{multline*}
Note the presence of $a_{1,1}$, which is $\kappa$ in the proposition, coming from the fact that the only constant terms in the right-hand side are $\kappa$ and $-a_{1,1}$, the constant term of $ -a_{1,1} \dfrac {\partial F_{1,\varnothing,\varnothing}}{\partial z}(-a_{1,0}\,z)$. Indeed, the constant terms of the other generating functions are null since the families counted by these generating functions do not include a diagram of size $1$.
\label{ex:h0,h1}
\end{ex}

\section{Combinatorial interpretation of the weight of a diagram}

The weighted generating function is not fully satisfying because of the dependence on the weight $w(C)$ which is not so easy to work with. In this section, we explain how to improve our model of diagrams to naturally include the weight coefficient. 

We distinguish the cases $s \geq 2$ and $s = 1$ for clarity purposes: some technical subtleties appear whenever $s=1$.  Here $s$ is again the insertion parameter from the Dyson-Schwinger equation and also appears in the definition of the weight.

\subsection{  Case $s \geq 2$  }

Our improved model of diagrams can be defined as follows.

\begin{defi}[$\omega_s$-marked diagram] For $s \geq 2$, a  diagram $C$ is \emph{$\omega_s$-marked} if it is decorated, connected and:
\begin{itemize}
\item its intervals may contain one or several crosses, named \emph{marks}, horizontally arranged;
\item for each chord $c$ of $C$, the intervals covered by $c$ (cf. Definition~\ref{def:covering}) contain $d(c)s-2$ marks in total.
\end{itemize}
\label{def:marked}
\end{defi}

An $\omega_2$-marked diagram is shown in \Cref{marked}. The notion of $\omega_s$-marked diagram obviously depends on the integer $s$. 
%For example, the diagram from \Cref{marked} cannot be changed into an $\omega$-marked diagram for $s=3$ since the intervals covered by chords $6$ and $7$, which are of decoration $1$, do not contain any mark (for $s=3$ there must be at least one mark below each chord).

\fig{[scale=1.7]{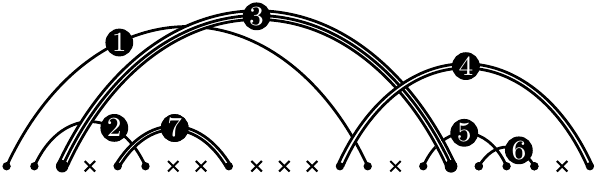}}{$\omega_2$-marked diagram}{marked}

\begin{proposition} Let $\mathcal A$ be a set of decorated connected diagrams. Then the weighted generating function of $\mathcal A$ is also the exponential generating function of $\omega_s$-marked diagrams of $\mathcal A$, i.e.
\[\sum_{C \in \mathcal A} w(C) \frac{z^{\|C\|}}{\|C\|!} = \sum_{n \geq 0} a_n \frac{z^n}{n!},\]
where $a_n$ is the number of marked diagrams of $\mathcal A$ of size $n$ (that is $\|\mathcal{A}\|=n$).
\label{prop:wgf}
\end{proposition}

\begin{proof}We have the equality
\[w(C) = \prod_{c=1}^{|C|} {d(c)s + \omega(c)- 2  \choose \omega(c) } = \prod_{c=1}^{|C|} {\omega(c)+1 + d(c)s - 2  - 1 \choose d(c)s - 2 } = \prod_{c=1}^{|C|} \multiset {\omega(c)+1} {d(c)s-2},\]
where $\multiset n k$ is the multiset coefficient, i.e. the number of multisets of size $k$ with elements taken from a set of size $n$. In other words, for $s=2$, the integer $w(C)$ counts the number of ways to insert $d(c)s-2$ marks into the intervals covered by $c$ (there are $\omega(c) + 1$ such intervals), for every chord $c$ of $C$. Therefore,  $w(C)$ counts the number of ways for $C$ to become $\omega_s$-marked. This proves the proposition.
\end{proof}

To have a recursive description of our $\omega_s$-marked diagrams, we need to extend the definition of an interval:
\begin{defi}[gap] A \textit{gap} in an $\omega_s$-marked diagram is a space between two dots/two marks/a mark and a dot.
\end{defi}

For example, the diagram of \Cref{marked} has 21 gaps. It turns out that the number of gaps is constant when the size of the $\omega_s$-marked diagram is fixed.

\begin{proposition} Given $s \geq 2$, every $\omega_s$-marked diagram $C$ has $(sn-1)$ gaps, where $n=\|C\|$.
\label{prop:nb_gaps}
\end{proposition}

\begin{proof}
 Let $\gamma$ be the number of gaps, and $m$ be the number of marks. Without any mark, a diagram with $|C|$ chords has 
$2|C|-1$ intervals, hence $2|C|-1$ gaps. Moreover, adding one mark to a marked diagram increases the number of gaps by one. Thus, we have the equality
\begin{equation}
\gamma = 2|C|-1 + m. \label{eq:nbgaps}
\end{equation} 
Furthermore, we have
\[ m = \sum_{c = 1}^{|C|} (s d(c) - 2) = - 2|C| + s \sum_{c = 1}^{|C|} d(c) = - 2|C| + sn.   \]
We deduce $\gamma = sn - 1$.
\end{proof}

%\begin{proposition} Let $s$ be an integer $>1$. If we consider that marks contribute to intervals, then an $\omega$-marked diagram of size $n$ has always $sn - 1 $ intervals.
%\end{proposition}
%
%\begin{proof} 
%The combination of the two above equalities concludes the proof.
%\end{proof}

The most natural way to build a larger diagram from an existing $\omega_s$-marked chord diagram is to add a root chord. This operation is not unique. For the rest of this subsection, we would like to study how many ways there are to insert a new root chord.
We first introduce a useful notation.

\begin{defi}[$\Gamma_j \mathcal A$]
Let $\mathcal A$ be a set of $\omega_s$-marked diagrams and $j$ be a positive integer. We denote by $\Gamma_j \mathcal A$ the set of $\omega_s$-marked diagrams $C$ such that the root chord has decoration $j$, and such that the deletion of the root chord induces one connected component, which must belong to $\mathcal A$. 
\label{def:Gammaj}
\end{defi}

The next lemma, which will be crucial for the rest of the document, claims that there is a constant number of ways to insert a new root chord into an $\omega_s$-marked diagram, when the size is fixed.  Recall that the size of a decorated chord diagram $C$ is $\|C\|$, that is, the sum of the decorations.

\begin{lem} 
For every set $\mathcal A$ of $\omega_s$-marked diagrams of size $n$, and for every positive integer $j$, the function ``\textit{deleting the root chord}'' is a $(sn-1)$-to-one function from $\Gamma_j \mathcal A$ to $\mathcal A$.
\label{lem:deletingroot}
\end{lem}

This lemma also holds for $s=1$ (the proof in that case is postponed to the next subsection)

\fig{[scale=1.7]{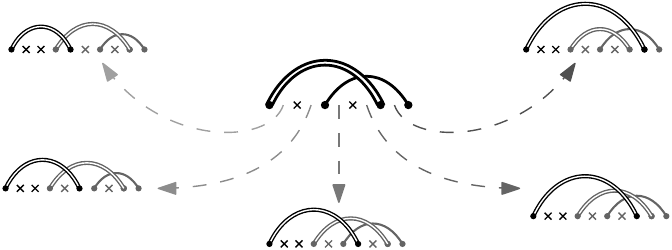}}{Given an $\omega_2$-marked diagram $C$ of size $n=3$, there are five $\omega_2$-marked diagrams such that the root chord has decoration $j=2$, and such that $C$ is recovered by deleting the root chord. Five is also the number of gaps of $C$.}{rootinsertion}

\begin{proof} (Case $s \geq 2$)

We prove that every diagram $C$ from $\mathcal A$ has as many preimages in $\Gamma_j \mathcal A$ under the function ``\textit{deleting the root chord}'' as gaps. Given a gap $g$, we construct a new $\omega_s$-marked diagram by inserting into $g$ the right endpoint of a new root chord of decoration $j$. We then put $js-2$ marks into the leftmost interval. Let us denote by $C_g$ the resulting diagram. Examples of $C_g$ are shown in Figure~\ref{rootinsertion}.

Let us now show that the marks put into $C_g$ make it $\omega_s$-marked. 

First, notice that the only interval covered by the root chord of $C_g$ is the leftmost one. Indeed, the root chord of $C_g$ is chord $1$ for the intersection order, so when we compute the covering of the chords of this new diagram (see Definition~\ref{def:covering}), all labels $1$ which are not in this first interval will be replaced by other labels. Since we placed $js-2$ marks into the only interval covered by the root chord,  $C_g$ is well $\omega_s$-marked at least as regards the root chord.

As for every non-root chord of $C_g$, the number of marks contained in the covered intervals does not change. Indeed, adding a root chord may add a new interval into the covering of a chord, but this newly formed interval will inherit the label of this chord, and some marks that were in the split interval. 

The diagram $C_g$ is therefore $\omega_s$-marked. By construction, it is also in $\Gamma_j \mathcal A$.
Moreover, for two gaps $g \neq g'$, the diagrams $C_g$ and $C_{g'}$ cannot be equal since the total number of marks and dots below the root chord of $C_g$ strictly increases with respect to the position of $g$. 

Thus there are as many gaps as preimages of $C$ under the function which deletes the root chord. Since by Proposition~\ref{prop:nb_gaps} there are $(sn-1)$ gaps in every diagram of $\mathcal A$, the lemma is proved.
\end{proof}

\subsection{The case $s=1$}

The case $s=1$ is special because the number of marks ``$d(c)s-2$'' can be negative when $d(c) = 1$. This explains the following alternative definition.

\begin{defi}[$\omega_1$-marked diagram] A  diagram $C$ is \emph{$\omega_1$-marked} if it is decorated, connected, contains marks in its intervals and:
\begin{itemize}
\item for each chord $c$ with decoration $\geq 2$, the intervals covered by $c$ exactly contain $d(c)-2$ marks in total.
\item for each chord $c$ with decoration $1$, the covering number of $c$ must be $0$, and the only interval covered by $c$ does not contain any mark.
\end{itemize}
\label{def:marked}
\end{defi}

Proposition~\ref{prop:wgf} also holds for $s=1$ with this notion of $\omega_1$-marked diagram.

\begin{proof}[Proof of Proposition~\ref{prop:wgf} for $s=1$] Remember that
\[w(C) = \prod_{c=1}^{|C|} {d(c) + \omega(c)- 2  \choose \omega(c) }.\]
For $d(c)>1$ the binomial is still a multiset coefficient, which counts the number of ways to insert $d(c)-2$ marks into the intervals covered by $c$. As for every chord $c$ such that $d(c)=1$, observe that
\[{d(c) + \omega(c)- 2  \choose \omega(c) } = \left\{ \begin{array}{ll}{-1 \choose 0} = 1 & \textrm{if }\omega(c)=0 \\ {\omega(c)-1 \choose \omega(c)} = 0 & \textrm{if }\omega(c)>0
\end{array}\right..\]
Thus, the weight $w(C)$ is non-zero if and only if every chord with decoration $1$ covers only one interval.
\end{proof}

As we said above, Lemma \ref{lem:deletingroot} holds also for $s=1$, but it is no longer the case for Proposition~\ref{prop:nb_gaps}. Indeed, $\|C\|-1$ is not equal to the number of gaps of a diagram $C$ (take for example the one-chord diagram which has size $1$ and $1$ gap). Nonetheless $\|C\|-1$ is still the number of preimages of any diagram $C$ under the function ``\textit{deleting the root chord}''.

\begin{proof}[Proof of  Lemma \ref{lem:deletingroot} for $s=1$]

Let $C$ be an $\omega_1$-marked diagram of size $n$, and let us count the number of its preimages.

Let $D$ be an $\omega_1$-marked diagram such that the removal of the root chord gives $C$. Let $g$ be the gap of $C$ where the right  endpoint of the root chord of $D$ is.

For shortness reasons, let us say that a gap is \textit{of decoration $d$} if it is in an interval covered by a chord of decoration $d$.

Gap $g$ cannot be of decoration $1$.  Indeed, the covering number of the chord that covers $g$ is at least equal to $1$, since the interval in which the root chord is inserted is split in two but $D$ is also $\omega_1$-marked. So necessarily $g$ is of decoration $\geq 2$.
Conversely, if $g$ is a gap of decoration $\geq 2$, one can construct a preimage $D$ by inserting a root chord whose right endpoint lies into $g$.

So the number of preimages is equal to the number of gaps of $C$ of decoration $\geq 2$. Clearly we have
\[
(\textrm{number of gaps of deco.} \geq 2)  = \textrm{number of gaps} - (\textrm{number of gaps of deco. }  1).
\]
But the number of gaps is equal to $2 |C| - 1 + \textrm{number of marks}$ (see \eqref{eq:nbgaps}), and the number of gaps of decoration $1$ is exactly equal to the number of \textit{chords} of decoration $1$ (because such chords are of covering number $0$ and contain no marks), hence:
\[
(\textrm{number of gaps of deco.} \geq 2)  = 2 |C| - 1 + \textrm{number of marks} - (\textrm{number of chords of deco. }  1)
\]
By definition of $\omega_1$-marked diagrams, the number of marks can be deduced via the following chain of equalities
\begin{align*}
\textrm{number of marks} & =  \sum_{\substack{ \textrm{chord }c \\ \textrm{of decoration }\geq 2}} d(c) - 2 \\ \ & =   \textrm{number of chords of deco. }  1 + \sum_{\textrm{chord }c}d(c) - 2 \\
\ & =  \textrm{number of chords of deco. }  1 + \|C\| - 2 |C|.
\end{align*}
By combining all the equalities, we deduce that the number of preimages is indeed $\|C\| - 1$.
\end{proof}

\section{Computation of the leading-log expansions}

All the mathematical analysis is based on Theorem~\ref{theo:intermsofwgf}, which expresses the N${}^k$LL expansions in terms of the weighted generating functions of diagrams with a given type. In this section, we compute the weighted generating functions involved in the cases where $k$ is small. 
Even if we have automated the computation of the weighted generating function of any type, we have chosen to begin by some simple examples, which we are going to generalize as we go along.

\subsection{Some examples for $s=2$ (Yukawa theory) }

The $s=2$ case corresponds to the fermion propagator in Yukawa theory as well as other propagators with the same combinatorics of their diagrammatics, as discussed in the background section.  

\subsubsection{Leading-log expansion}

Thanks to Example~\ref{ex:h0,h1}, we know that $H_0(z) = -F_{1,\varnothing,\varnothing}(-a_{1,0}\,z)$, where $F_{1,\varnothing,\varnothing}$ is the generating function of connected diagrams with only one terminal chord and no chord with decoration $\geq 2$.

\begin{proposition} The weighted generating function of diagrams with only one terminal chord and no chord with decoration $\geq 2$ is equal to
\[F_{1,\varnothing,\varnothing}(z)=1-\sqrt{1-2z}.\]
Consequently, the leading-log expansion for $s=2$ is $H_0(z) = \sqrt{1+2\,a_{1,0}\,z}-1$.
\label{prop:lls2}
\end{proposition}

\begin{proof} 
 We wish to enumerate the connected diagrams with only one terminal chord (no decorations and no marks).
The study of this has already been done in \cite{CYchord}, which we explain here in terms of Lemma~\ref{lem:deletingroot}. 
%We denote by $f_n$ the number of connected diagrams of $n$ chords with only one terminal chord. Obviously, we have $f_1 = 1$ since there is only one diagram of size $1$.

Let $\mathcal F_n$ be the set of connected diagrams of size $n$ with only one terminal chord. Consider $C$ a diagram of $\mathcal F_{n+1}$ with $n \geq 1$. Remove the root chord of $C$: there results one or several connected components $C_1, \dots, C_r$. But each of these components must contain at least one terminal chord (the last chord for the intersection order). So $r$ must be equal to $1$, and $C_1$ must have $n$ chords. In other words, $\mathcal F_{n+1}$ is exactly $\Gamma_1 \mathcal F_{n}$ (recall  Definition~\ref{def:Gammaj}).

We can then use Lemma~\ref{lem:deletingroot} to prove that, if $f_n$ denotes the number of connected diagrams of $n$ chords with only one terminal chord, then $f_{n+1} = (2n-1) f_n$. 
%
%Reciprocally, given a connected diagram $C_1$ of size $n-1$, there exist $2n-3$ ways to construct a diagram of size $n$ from $C_1$ by adding a root chord. This number corresponds to the number of intervals in $C_1$. \Cref{oneterminalchord} illustrates that fact. This shows the equality $f_n = (2n-3) f_{n-1}$.
%
%
%\fig{[scale=1.7]{oneterminalchord}}{Five connected chord diagrams of size $4$ can be built by adding a root chord into a connected chord diagram of size $3$. }{oneterminalchord}
%

Therefore, since $f_1=1$, we have $f_n  = (2n-3) \times (2n-5) \times \dots \times 3 \times 1 = (2n-3)!!$, which implies $F_{1,\varnothing,\varnothing}(z)= \sum_{n \geq 1} \frac{(2n-3)!!}{n!} z^n =  1-\sqrt{1-2z}$.
\end{proof}

\subsubsection{Next-to leading-log expansion}

As for the next-to leading log expansions, we saw that $H_1(z)$ is a linear combination of the derivatives of
$F_{1,\varnothing,\varnothing}$, $F_{2,\varnothing,\varnothing}$, $F_{1,  \{(1,1)\} ,\varnothing}$ and $ F_{1,\varnothing,\{2\}}$.

\begin{proposition} We have
\[F_{2,\varnothing,\varnothing}(z)=1-\sqrt {1-2\,z}-z,\]
\[F_{1,  \{(1,1)\} ,\varnothing}(z) =  F_{1,\varnothing,\{2\}}(z) = 
\frac 1 2 \ln  \left( 1-2\,z \right) \sqrt {1-2\,z
}+z
\]
Therefore, by Theorem~\ref{theo:intermsofwgf}, the next-to leading-log expansion for $s=2$ is
\[H_1(z) = {\frac {- \left( a_{{2,0}}+a_{{1,1}}a_{{1,0}} \right)  \,\ln   \left( 1+2\,za_{{1,0}} \right) 
 }{2 a_{{1,0}}\, \sqrt {1+2\,za
_{{1,0}}} }}\]
\label{prop:ntlle}
\end{proposition}

\begin{proof}
Let us begin by $F_{2,\varnothing,\varnothing}$, which is the generating function of $\omega_2$-marked diagrams with one terminal chord, which has decoration $2$, and where every other chord is not terminal and has decoration $1$. Let $\mathcal F^{{2,\varnothing,\varnothing}}_n$ be the class of such diagrams with size $n$.

When $n > 2$, the root chord of every diagram from $\mathcal F^{{2,\varnothing,\varnothing}}_n$ is not terminal (otherwise the diagram would not be connected). So it has decoration $1$. When we remove it, the resulting diagram belongs to  $\mathcal F^{{2,\varnothing,\varnothing}}_{n-1}$ by the same argument as the proof of the previous proposition. So $\mathcal F^{{2,\varnothing,\varnothing}}_n = \Gamma_1 {\mathcal F^{{2,\varnothing,\varnothing}}_{n-1}}$, and by Lemma \ref{lem:deletingroot}, the cardinality of $\mathcal F^{{2,\varnothing,\varnothing}}_n$ should satisfy 
\[  |\mathcal F^{{2,\varnothing,\varnothing}}_n| = (2n-3) |\mathcal F^{{2,\varnothing,\varnothing}}_{n-1}|.\]
But $\mathcal F^{{2,\varnothing,\varnothing}}_2$ contains only the diagram with only one chord of decoration $2$. Thus $ |\mathcal F^{{2,\varnothing,\varnothing}}_n| = (2n-3)!!$ for $n \geq 2$. So there are as many diagrams from $\mathcal F^{{2,\varnothing,\varnothing}}_n$ as non-decorated connected diagrams of size $n$ with only one terminal chord except for $n=1$ (there is no diagram in $\mathcal F^{{2,\varnothing,\varnothing}}_1$ since the decorated chord makes the size at least equal to $2$; and there is only one diagram of size $1$ with one terminal chord: the one-chord diagram). This explains why \[F_{2,\varnothing,\varnothing}(z) = F_{1,\varnothing,\varnothing}(z) - z = 1-\sqrt {1-2\,z}-z.\]

The computation of $F_{1,  \{(1,1)\} ,\varnothing}(z)$ was also done in \cite{CYchord}, but let us explain the proof again.
Let $\mathcal F^{1,  \{(1,1)\},\varnothing}_n$ be the class of diagrams counted by $F_{1,  \{(1,1)\} ,\varnothing}(z)$, i.e. the diagrams made of $n$ chords of decoration $1$,  including exactly two terminal chords, which are last and next-to-last for the intersection order. 

Consider $n \geq 3$ and $C$ a diagram from $\mathcal F^{1,  \{(1,1)\},\varnothing}_n$. Removing the root chord of $C$ leads to two possibilities, which are illustrated in Figure~\ref{twoterminalchords}:
\begin{enumerate}
\item The resulting diagram is connected. The class of such diagrams is $\Gamma_1  \mathcal F^{1,  \{(1,1)\},\varnothing}_{n-1}$ (see Definition~\ref{def:Gammaj}).
\item The resulting diagram has several connected components $C_1, \dots, C_r$, nested in that order from top to bottom. Since there are only two terminal chords and each $C_i$ contains at least one terminal chord, we have $r=2$. Moreover, since the terminal chord of $C_1$ must be at position $n-1$ for the intersection order of $C$, the diagram $C_2$ must have size $1$, i.e. $C_2$ is just constituted of one chord. So we can map $C$ to another diagram $C'$ obtained from $C$ by removing $C_2$. The diagram $C'$ has size $n-1$ and has only one terminal chord, so belongs to what we called $\mathcal F_{n-1}$ in the previous proof (the class of diagrams of size $n-1$ with only one terminal chord). One can easily recover $C$ from $C'$ by adding a chord around the right endpoint of the root chord of $C$.

\fig{[scale=1.2]{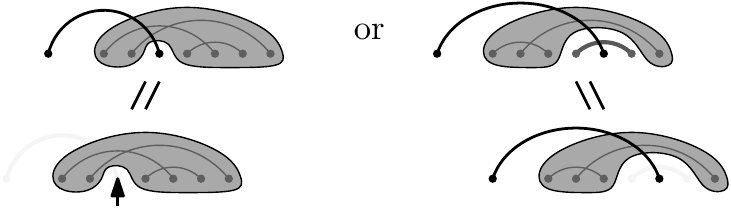}}{The two possible decompositions for a diagram in $\mathcal F^{1,  \{(1,1)\},\varnothing}_n$ }{twoterminalchords}

We thus prove that there is a bijection between diagrams from $\mathcal F^{1,  \{(1,1)\},\varnothing}_n$  satisfying this point and diagrams from $\mathcal F_{n-1}$. 
\end{enumerate}
At this point, we have proved the bijection
\[  \mathcal F^{1,  \{(1,1)\},\varnothing}_n \ \approx_{bij} \  \Gamma_1 \mathcal F^{1,  \{(1,1)\},\varnothing}_{n-1} \ \  \cup \  \    \mathcal F_{n-1}.\]
In terms of cardinality, it means
\[  \left|\mathcal F^{1,  \{(1,1)\},\varnothing}_n\right| = (2n - 3) \times \left|\mathcal F^{1,  \{(1,1)\},\varnothing}_{n-1}\right|  + (2n-5)!!
\]
(we used Lemma~\ref{lem:deletingroot}) which can be rewritten as
\[   \left|\mathcal F^{1,  \{(1,1)\},\varnothing}_n\right| \, \frac{ z^{n-1}}{(n-1)!} = 2 z \left|\mathcal F^{1,  \{(1,1)\},\varnothing}_{n-1}\right|  \frac{z^{n-2}}{(n-2)!} - \left|\mathcal F^{1,  \{(1,1)\},\varnothing}_{n-1} \right|  \, \frac{ z^{n-1}}{(n-1)!} + (2n-5)!!  \, \frac{ z^{n-1}}{(n-1)!}. \]
When we sum over $n \geq 2$, we recognize the generating function equality
\[
 \frac {\partial}{\partial z} \left( F_{1,  \{(1,1)\},\varnothing}(z) \right) = 2z  \frac {\partial}{\partial z}   \left( F_{1,  \{(1,1)\},\varnothing}(z) \right) - F_{1,  \{(1,1)\},\varnothing}(z) + 1-\sqrt {1-2\,z} - z.
\]
This differential equation can be straightforwardly solved by
\[  F_{1,  \{(1,1)\},\varnothing}(z)  = \frac 1 2 \ln  \left( 1-2\,z \right) \sqrt {1-2\,z
}+z. \]

Finally, consider $\mathcal F^{1, \varnothing,\{2\}}_n$ the class of diagrams counted by $F_{1, \varnothing,\{2\} }(z)$, i.e. the $\omega_2$-marked diagrams of size $n$ with only one terminal chord and a chord -- different from the terminal chord -- of decoration $2$.

The recursion here is simple for diagrams in  $\mathcal F^{1, \varnothing,\{2\}}_n$ : either the root chord has decoration $1$, or it has decoration $2$. In both cases, the removal of the root chord induces one single connected component  $C'$ (only one because there is only one terminal chord). When the decoration of the root chord is $1$, the diagram $C'$ belongs to  $\mathcal F^{1, \varnothing,\{2\}}_{n-1}$; when it is $2$, the diagram $C'$ belongs to $\mathcal F_{n-2}$.
That is why we obtain the bijection for $n \geq 2$
\[  \mathcal F^{1, \varnothing,\{2\}}_n \ \approx_{bij} \  \Gamma_1 \mathcal F^{1, \varnothing,\{2\}}_n  \ \  \cup \  \    \Gamma_2 \mathcal F_{n-2}  \]
Moreover, thanks to Lemma~\ref{lem:deletingroot}, we see that there is a bijection between $\Gamma_2 \mathcal F^{1, \varnothing,\varnothing}_{n-2}$ and $\Gamma_1 \mathcal F^{1, \varnothing,\varnothing}_{n-2}$. But
we saw in the proof of Proposition~\ref{prop:lls2} that $\Gamma_1 \mathcal F_{n-2}$ is $\mathcal F_{n-1}$. Therefore
\[  \mathcal F^{1, \varnothing,\{2\}}_n \ \approx_{bij} \  \Gamma_1 \mathcal F^{1, \varnothing,\{2\}}_n  \  \cup \     \Gamma_2 \mathcal F_{n-2} \ \approx_{bij} \  \Gamma_1 \mathcal F^{1, \varnothing,\{2\}}_n  \  \cup \    \mathcal F_{n-1} \  \approx_{bij} \   \mathcal F^{1,  \{(1,1)\},\varnothing}_n. \]
In other words, $\mathcal F^{1, \varnothing,\{2\}}_n$ and $ \mathcal F^{1,  \{(1,1)\},\varnothing}_n$ are in bijection, which implies the equality of the generating functions.
\end{proof}

The next-to next-to leading-log expansion is given in Proposition~\ref{NNLLs2}, after the general theory which can be used to obtain any next-to${}^k$ leading-log expansion is developed.

\subsection{Three lemmas}

The objective of this subsection is to establish some lemmas which are useful in the generic procedure (described in the next subsection) that computes the generating functions of $\omega_s$-marked diagrams with respect to their types.

As previously said, the most natural way to grow diagrams from a given set of $\omega_s$-diagrams is to insert one or several root chords into these diagrams (see Figure~\ref{rootchordsinsertion} for an illustration). It turns out that the operation of inserting an arbitrary number of root chords translates nicely to generating functions.

\fig{[width = 0.7 \textwidth]{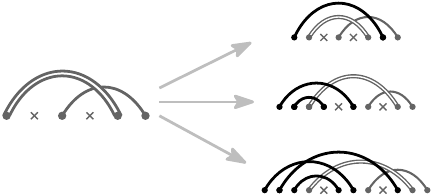}}{The three diagrams on the right are examples of diagrams obtained from the diagram on the left by root chord insertions. }{rootchordsinsertion}

\begin{lem}
Let  $\mathcal A$ be a set of  $\omega_s$-marked  diagrams whose exponential generating function is $A(z)$. 

If $\mathcal B$ denotes the set of diagrams obtained from $\mathcal A$ by successive insertions of root chords of decoration $1$ (any number of insertions is allowed, including $0$ and $1$), then the exponential generating function of $\mathcal B$ is 
\[B(z) = (1-sz)^{1/s} \left(\int_{0}^z (1-st)^{-\frac{s+1} s} A'(t) dt\right). \]
\label{lem:operator}
\end{lem}
\begin{proof}
Let $a_n$ and $b_n$ be the number of elements  in $\mathcal A$ and $\mathcal B$ (respectively) of size $n$. Two possibilities can occur for a diagram $C$ of size $n+1$ in $\mathcal A$: either $C$ belongs to $\mathcal B$; or it comes from an insertion of a root chord of an element of $\mathcal A$ of size $n$. By Lemma~\ref{lem:deletingroot}, we deduce
\[b_{n+1}  = a_{n+1} + (sn-1) b_{n}.\]
We transform this equality into
\[\sum_{n} b_{n+1} \frac{z^n}{n!} = \sum_{n} a_{n+1} \frac{z^n}{n!} + s \times \sum_{n} b_{n} \frac{z^n}{(n-1)!} - \sum_{n} b_{n} \frac{z^n}{n!}.\]
This is the generating function equality
\[B'(z) = A'(z) + sz B'(z) - B(z).\]
By multiplying each side of this formula by $(1-sz)^{-(s+1)/s}$, we can rewrite this as
\[(1-sz)^{-1/s} B'(z) + (1-sz)^{-1/s - 1} B(z) = (1-sz)^{-(s+1)/s} A'(z). \]
We recognize the derivative of $(1-sz)^{-1/s} B(z)$ on the left side. By integrating the equation from $0$ to $z$ (plugging $z=
0$ into $(1-sz)^{-1/s} B(z)$ gives $0$ because $B(0)=0$), we recover the equation of the statement.\end{proof}

To use the previous lemma for the computation of the generating function of diagrams of a given type, we have to know what the ground set of diagrams into which we insert root chords is. This is the subject of the next lemma.

\begin{lem} \label{lem:peeling_root_chords}
Let $\mathcal F_{\Delta,\mathcal G,\mathcal D}$ be the set of $\omega_s$-marked diagrams of type $(\Delta,\mathcal G,\mathcal D)$.

Every diagram of $\mathcal F_{\Delta,\mathcal G,\mathcal D}$ 
can be obtained by successive insertions of root chords of decoration $1$ into an $\omega_s$-marked diagram $C$, where $C$ satisfies one of the three following (disjoint) possibilities:
\begin{enumerate}
\item $C$ has only one chord, and this chord has decoration $\Delta$ (only possible if $\mathcal G = \varnothing$ and $\mathcal D = \varnothing$);
\item There exists $d \in \mathcal D$ such that $C \in \Gamma_d \mathcal F_{\Delta,\mathcal G,\mathcal D \backslash \{d\}}$ (see Definition~\ref{def:Gammaj}) (only possible if $\mathcal D \neq \varnothing$);
\item The removal of the root chord of $C$ disconnects the diagram (only possible if $\mathcal G \neq \varnothing$). \label{case:disconnecting}
\end{enumerate}
\end{lem}

\begin{proof} Straightforward. If a diagram does \textit{not} belong to  $\Gamma_1 \mathcal F_{\Delta,\mathcal G,\mathcal D}$, then it must satisfy one of these three possibilities.
\end{proof}

\begin{ex} The two previous lemmas enable the computation of simple generating functions. Let us calculate the generating function of diagrams with type $(1,\varnothing,\varnothing)$ for any $s\geq 2$.  By Lemma~\ref{lem:peeling_root_chords}, the set of such diagrams is also the set obtained by inserting an arbitrary number of root chords into the diagram with only one chord, which has decoration $1$. So, by Lemma~\ref{lem:operator} (here $A(z) = z$), its generating function is 
\[ F_{1,\varnothing,\varnothing}(z) = (1-sz)^{1/s} \left(\int_{0}^z (1-st)^{-\frac{s+1} s} dt\right) = 1 - (1-sz)^{1/s}. \]
If we translate this expression in terms of coefficients, we see that the number of diagrams of type $(1,\varnothing,\varnothing)$ and of size $n$ is equal to $(s(n-1)-1) \times (s(n-2) - 1) \times (s(n-3) - 1) \times \dots \times (s - 1)$ (which corresponds to the double factorial $(2n-3)!!$ for $s=2$, and the triple factorial $(3n-4)!!!$ for $s=3$).
For $s=2$, we recover Proposition~\ref{prop:lls2}.

We can also compute the generating function of diagrams with type $(1,\varnothing,\{2\})$ for any $s \geq 2$. Indeed, by Lemma~\ref{lem:peeling_root_chords}, these diagrams can be obtained by root chord insertions from diagrams of $\Gamma_2 \mathcal F_{1,\varnothing,\varnothing}$. To apply  Lemma~\ref{lem:operator}, we need to compute the generating function of the latter diagrams. To do so, let us denote $f_{n}$ the number of diagrams of type $(1,\varnothing,\varnothing)$ and of size $n$. By Lemma~\ref{lem:deletingroot}, since the root chord has decoration $2$, the number of diagrams of size $n$ in $\Gamma_2 \mathcal F_{1,\varnothing,\varnothing}$ is equal to $(s(n-2)-1) f_{n-2}$. So the generating function of such diagrams is
\[ \sum_{n \geq 3} (s(n-2)-1) f_{n-2} \frac{z^n}{n!} = s \sum_{n \geq 3} n f_{n-2} \frac{z^n}{n!} - (2s+1) \sum_{n \geq 3}f_{n-2} \frac{z^n}{n!}.\]
We recognize from the right the function $s z \int F_{1,\varnothing,\varnothing} dz - (2s+1) \int \int F_{1,\varnothing,\varnothing} dz^2$ where $\int y\  dz$ and $\int \int y\  dz^2$ respectively denote the antiderivative and the second antiderivative of a function $y(z)$ (Leibniz's notation). We find after calculations that the generating function of  $\Gamma_2 \mathcal F_{1,\varnothing,\varnothing}$ is 
\[\frac{(1 - s z)^{\frac 1 s + 1}}{s + 1} - \frac{ z^2 } 2 + z - \frac 1 {s+1}. \]
So by Lemma~\ref{lem:operator}, we can deduce that 
\[F_{1,\varnothing,\left\{2\right\}}(z) = {\frac {s \left( s+z-2 \right) + \left( 1-sz \right) ^{1/s}
 \left(  \left( s-1 \right) \ln  \left( 1-sz \right) - \left( s-2
 \right) s \right) }{s \left( s-1 \right) }}.
 \]
\end{ex}

As these examples illustrate, the computations of the $F_{\Delta,\mathcal G, \mathcal D}$ functions  can be deduced from the enumeration of diagrams satisfying cases (1), (2) or (3) from Lemma~\ref{lem:peeling_root_chords}. But Case \eqref{case:disconnecting}, which occurs whenever $\mathcal G \neq \varnothing$, is substantially harder to understand. The next lemma establishes a bijection  from the set of diagrams satisfying Case \eqref{case:disconnecting}, which facilitates their enumeration.

\begin{lem} Let $\mathcal F_n^{\Delta,\mathcal G,\mathcal D}$ be the set of $\omega_s$-marked diagrams of size $n$ and of type $(\Delta,\mathcal G,\mathcal D)$.

The subset of $\mathcal F_n^{\Delta,\mathcal G,\mathcal D}$ formed with diagrams such that the removal of the root chord of $C$ disconnects the diagram (case \eqref{case:disconnecting} from Lemma~\ref{lem:peeling_root_chords})  are in bijection with the set of pairs $(C_1,C_2)$, such that:
\begin{enumerate}
\item there exists a decomposition of $\mathcal G$ such that $\mathcal G = \{(\delta,\gamma)\} \cup \mathcal G_1 \cup \mathcal G_2$ \label{item1};
\item there exists a decomposition of $\mathcal D$ such that $\mathcal D = \mathcal D_1 \cup \mathcal D_2$ \label{item2};
\item $C_1 \in \mathcal F_{n-m}^{\Delta,\mathcal G_1,\mathcal D_1}$ \label{item3};
\item $C_2 \in \Gamma_1 \mathcal F_{m}^{\delta,\mathcal G_2,\mathcal D_2}$ \label{item4};
\item $m = \gamma  + \delta - 1 + \sum_{(d,g) \in \mathcal G_2} (d+g-1) + \sum_{d \in \mathcal D_2} (d-1)$ \label{item5};
\item $\gamma =  1 \Rightarrow  \mathcal G_2 = \mathcal D_2 = \varnothing$ .\label{item6}
\end{enumerate}
%
%where $C_1 \in \mathcal F_{n-m}^{(\Delta,\mathcal G_1,\mathcal D_1)}$, $C_2 \in \mathcal F_{m}^{(\delta,\mathcal G_2,\mathcal D_2)}$, there exist a decomposition of $\mathcal G$ into disjoint parts $\mathcal G = \{(\delta,\gamma)\} \cup \mathcal G_1 \cup \mathcal G_2$ and a decomposition of $\mathcal D$ into disjoint parts $\mathcal D = \mathcal D_1 \cup \mathcal D_2$ such that $C_r$ belongs to $\mathcal F_{(\delta,\mathcal G_1,\mathcal D_1)}$, and such 
\label{lem:when_disconnection}
\end{lem}

\fig{[width=0.9 \textwidth]{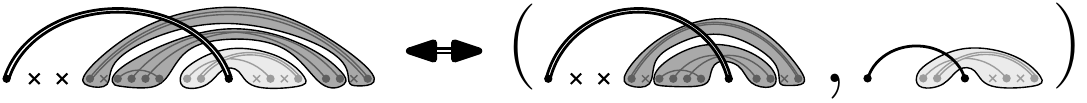}}{Illustration of Lemma~\ref{lem:when_disconnection}. }{disconnection}

\begin{proof}Let $C$ be a diagram from $\mathcal F_n^{\Delta,\mathcal G,\mathcal D}$ such that the removal of the root chord disconnects $C$ into several connected $\omega_s$-marked components $D_1, D_2, \dots, D_r$. We assume that the diagrams $D_i$ are sorted from top to bottom, and each of these diagrams inherits the marks inside the intervals covered by its chords. We denote $C_1$ the diagram obtained from $C$ by deleting $D_r$, and $C_2$ the diagram obtained from $C$ by deleting $D_1, D_2, \dots, D_{r-1}$ and setting the decoration of the root chord to be $1$ (see Figure~\ref{disconnection} for an example).

Let us prove the six items of the lemma.

Whenever the root chord of $C_2$ (which had decoration $1$) is removed, by construction, there remains only one component, which is $D_r$. Let $m$ be the size of $D_r$, and $(\delta,\mathcal G_2,\mathcal D_2)$ be the type of $D_r$. We thereby recover that $C_2 \in \Gamma_1 \mathcal F_{m}^{\delta,\mathcal G_2,\mathcal D_2}$, i.e item \eqref{item4}.

Since it is the union of $C_1$ and $D_r$, the diagram $C$ has as its size $m$ \textit{plus} the size of $C_1$. That is why the size of $C_1$ is $n-m$. Moreover, the first terminal chord of $C_1$ is inside $D_1$ so it is also the first terminal  of $C$. Therefore it has decoration $\Delta$. We set $(\Delta,\mathcal G_1,\mathcal D_1)$ to be the type of $C_1$. Thus item \eqref{item3} holds.

Let us prove now item \eqref{item2}, that is to say $\mathcal D = \mathcal D_1 \cup \mathcal D_2$. Since chords of $D_r$ never intersect chord of $C_1$ (the root chord of $C$ excluded), the set of chords which are not terminal in $D_r$ or $C_1$ is also the set of the non-terminal chords in $C$. Moreover, non-terminal chords carries the same decorations between $C$ and $C_1 \cup D_r$. This explains $\mathcal D = \mathcal D_1 \cup \mathcal D_2$.

Let us go on with item \eqref{item1}. Let $t'_1,t'_2,\dots,t_r'$ the positions of the terminal chords in $C_1$, and $t_1'',t_2'',\dots,t_s''$ their positions in $D_r$.
By definition of the intersection order (see Definition~\ref{def:intersection}), the chords of $C_1$ are smaller than all the chords of $D_r$, and their relative positions remain the same. So the positions of the terminal chords in $C$ are  $t'_1,t'_2,\dots,t_r', |C_1| + t_1'', |C_1| + t_2'',\dots, |C_1| + t_s''$ (recall that $|C_1|$ is the number of chords of $C_1$, not its size) in that order. By definition of $\mathcal G$, we then have \[\mathcal G = \mathcal\  G_1 \cup \left\{ (d(|C_1| + t_1''), t_1'' +  |C_1| - t_r' ) \right\} \  \cup \  \mathcal G_2.\]
The chord labeled by $|C_1| + t_1''$ in $C$ is the chord labeled by $t_1''$ in $D_r$. So $d(|C_1| + t_1'') = \delta$. Thus, by setting $ \gamma = t_1'' +  |C_1| - t_r'$, we recover item \eqref{item1}. Note that $\gamma = t_1''$ since $t'_r$  is  the label of the last terminal in $C_1$ so is equal to the number of chords in $C_1$ (the intersection order always finishes by a terminal chord).

As for item \eqref{item5}, we know that $t_s''$ is equal to the number of chords of $D_r$ (for the same reason as above). Moreover, the size of a diagram is equal to the number of chords plus the sum over all chords of the difference between their decorations and $1$. Hence
\[ m = t_s'' + \delta - 1 + \sum_{(d,g) \in \mathcal G_2} (d-1) + \sum_{d \in \mathcal D_2} (d-1).\]
But $t_s'' =  t_1'' + \sum_{i=2}^s (t_i'' - t_{i-1}'')  =  \gamma  +  \sum_{(d,g) \in \mathcal G_2} g.$ By combining both previous equations, we obtain the expression of $m$ given by item \eqref{item5}.

There remains item \eqref{item6} to prove. Assume that $\gamma = 1$. It implies that $t_1'' = 1$, which means that the root chord of $D_r$ is terminal. But a root chord is terminal in a connected diagram if and only if the diagram has only one chord. Thus we have $m = \delta$, $ \mathcal G_2 = \mathcal D_2 = \varnothing$.

For the reverse construction, consider $C_1 \in \mathcal F_{n-m}^{\Delta,\mathcal G_1,\mathcal D_1}$ and $C_2 \in \Gamma_1 \mathcal F_{m}^{\delta,\mathcal G_2,\mathcal D_2}$ satisfying the six above items. We can recover $C$ by the following construction: insert $C_2$ inside $C_1$ so that $C_2$ lies on the two intervals adjacent to the right endpoint of the root chord of $C_1$, and so that the right endpoint of the root chord of $C_2$ coincides with the right endpoint of the root chord of $C_1$. Then remove the root chord of $C_2$. The resulting diagram is $C$.
\end{proof}

%
%
%\begin{proposition}
%
%Let $f_n^{(\Delta,\mathcal G,\mathcal D)}$ be the number of $\omega$-marked diagrams of type $(\Delta,\mathcal G,\mathcal D)$.
%The following recurrence holds for $n > \Delta$:
%\begin{multline*}f_{n}^{(\Delta,\mathcal G,\mathcal D)} = (2n-3) f_{n-1}^{(\Delta,\mathcal G,\mathcal D)} + \sum_{d \in \mathcal D} (2n - 2d - 1) f_{n-1}^{(\Delta,\mathcal G,\mathcal D - \{d\})} \\+ \sum_{\substack{
%((\delta,\gamma),\mathcal G_1,\mathcal G_2)\textrm{ such that }\{(\delta,\gamma)\} \cup \mathcal G_1 \cup \mathcal G_2 = \mathcal G \\
%\textrm{and }(\gamma \neq 1 \textrm{ or } \mathcal G_2 = \varnothing )   \\
%(\mathcal D_1,\mathcal D_2)\textrm{ such that }\mathcal D_1 \cup \mathcal D_2 = \mathcal D 
%}} (2m-1) f_{m}^{(\delta,\mathcal G_2,\mathcal D_2)} f_{n-m}^{(\Delta,\mathcal G_1,\mathcal D_1)}
%\end{multline*}
%where $m$ in the sum is equal to 
%$\gamma  + \delta - 1 + \sum_{(d,g) \in \mathcal G_2} (d+g-1) + \sum_{d \in \mathcal D_2} (d-1)$.
%
%
%For small $n$, we have $f_{\Delta}^{(\Delta,\varnothing,\varnothing)} = 1$ and for every triple $(n,\mathcal G,\mathcal D)$ different from $(\Delta,\varnothing,\varnothing)$ such that $n \leq \Delta$ we have $f_{n}^{(\Delta,\mathcal G,\mathcal D)} = 0.$
%\end{proposition}

\subsection{Description of the generic method}
\label{ss:method}
This subsection presents a generic procedure that yields an explicit formula for the next-to$^k$ leading log expansion (for any $s$) via the computation of the exponential generating functions of $\omega_s$-marked diagrams with respect to their types (see Theorem~\ref{theo:intermsofwgf}). 
%By Theorem~\ref{theo:intermsofwgf}, this is sufficient to find the next-to$^k$ leading-log expansions $H_k(z)$. 
%To do so, we are going to find formulas characterising these generating functions. 

The procedure to compute the generating function $F_{\Delta,\mathcal G,\mathcal D }(z)$ of $\omega_s$-marked diagrams of type $(\Delta,\mathcal G,\mathcal D)$ is the following:
\begin{enumerate}
\item Using Lemma~\ref{lem:peeling_root_chords}, determine all the possible forms of diagrams which generate all diagrams of type $(\Delta,\mathcal G,\mathcal D )$ via successive insertions of root chords. We want  to compute the generating function $A(z)$ of such diagrams.
\begin{enumerate}
\item If $\mathcal D = \mathcal G = \varnothing$, the diagrams of type $(\Delta,\mathcal G,\mathcal D)$ are all generated from  the one-chord diagram of decoration $\Delta$ (case (1) of Lemma~\ref{lem:peeling_root_chords}). Then set $A(z) = z^\Delta/\Delta!$. Otherwise, start from $A(z) = 0$.
\item If $\mathcal D \neq \varnothing$, then we have to consider all diagrams satisfying (2) in Lemma~\ref{lem:peeling_root_chords}. In this case, for every $d \in \mathcal D$, we recursively compute $F_{\Delta,\mathcal G,\mathcal D \backslash \{d\}}$. The generating function we are looking for is $\sum_n (s(n-d) - 1) f_{n-d} \frac {z^n} {n!}$, where $f_{n-d}$ is the $(n-d)$ coefficient of $F_{\Delta,\mathcal G,\mathcal D \backslash \{d\}}$. One then has to add  the function
\[\sum_n (sn-sd- 1) f_{n-d} \frac {z^n} {n!} = s z \underbrace{\int \dots \int}_{(d-1)\textrm{ times}} F_{\Delta,\mathcal G,\mathcal D \backslash \{d\}} dz^{d-1} - (sd+1) \underbrace{\int \hspace{-5pt} \int \dots \int}_{d\textrm{ times}} F_{\Delta,\mathcal G,\mathcal D \backslash \{d\}} dz^{d}\]
to $A(z)$ for every $d \in \mathcal D$.
\item If $\mathcal G \neq \varnothing$, then we have to compute the generating function of diagrams in which the removal of the root chord disconnects the diagram. To do so, we use Lemma~\ref{lem:when_disconnection}:
we list all the decompositions $\mathcal G = \{(\delta,\gamma)\} \cup \mathcal G_1 \cup \mathcal G_2$ and  $\mathcal D = \mathcal D_1 \cup \mathcal D_2$ satisfying \eqref{item6} in Lemma~\ref{lem:when_disconnection}.
Note that there are finitely many such decompositions.  For every such decomposition, we recursively compute $F_{\Delta,\mathcal G_1,\mathcal D_1}$ and $F_{\delta,\mathcal G_2,\mathcal D_2}$. At this point, set $m = \gamma  + \delta - 1 + \sum_{(d,g) \in \mathcal G_2} (d+g-1) + \sum_{d \in \mathcal D_2} (d-1)$. Then we
want to find $\sum_n (sm-1) f_{m}^{\delta,\mathcal G_2,\mathcal D_2} f_{n-m}^{\Delta,\mathcal G_1,\mathcal D_1} \frac{z^n}{n!}$, where $f_k^{\Delta,\mathcal G,\mathcal D }$ denotes the number of diagrams of type $(\Delta,\mathcal G,\mathcal D )$. To do so, we compute the
 $m$th antiderivative of $F_{\Delta,\mathcal G_1,\mathcal D_1}$, then multiply it by the coefficient of $z^m/m!$ in $F_{\delta,\mathcal G_2,\mathcal D_2}$, and finally multiply it by $(sm - 1)$. As before, we add the resulting function to $A(z)$,  for every decomposition $\mathcal G = \{(\delta,\gamma)\} \cup \mathcal G_1 \cup \mathcal G_2$ and  $\mathcal D = \mathcal D_1 \cup \mathcal D_2$.
\end{enumerate}

\item $A(z)$ is finally the generating function of diagrams which generate the set of diagrams of type  $(\Delta,\mathcal G,\mathcal D)$ with root chord insertions. By Lemma~\ref{lem:operator}, $F_{\Delta,\mathcal G,\mathcal D}(z)$ is thus \[(1-sz)^{1/s} \left(\int_{0}^z (1-st)^{-\frac{s+1} s} A'(t) dt \right).\] 
\end{enumerate}

\begin{ex} Set $s=2$, $\Delta = 3$, $\mathcal G = \{ (2,1),(2,1) \}$ and $\mathcal D = \{4\}$. We have $\mathcal D \neq \varnothing$. One one hand, the generating function of  all diagrams satisfying (2) in Lemma~\ref{lem:peeling_root_chords} is given by
\[ 2 z \, \int \hspace{-3pt} \int \hspace{-3pt} \int  F_{3,\{ (2,1),(2,1) \},\varnothing}(t) dt^3 - 9 \int \hspace{-3pt} \int \hspace{-3pt} \int \hspace{-3pt} \int  
F_{3,\{ (2,1),(2,1) \},\varnothing}(t) dt^4.\]
One the other hand, diagrams of type $(3,\{ (2,1),(2,1) \},\{4\})$ such that the removal of the root chord disconnects the diagrams are all given by Lemma~\ref{lem:when_disconnection} by $(\delta,\gamma) = (2,1)$, $\mathcal G_1 = \{ (2,1) \}$, $\mathcal D_1 = \{4\}$, and $\mathcal G_2 = \mathcal D_2 = \varnothing$ (remember item \eqref{item6} from the same lemma). So knowing that $F_{\delta,\mathcal G_2,\mathcal D_2} = z^2/2!$, and $m=2$, the generating function of diagrams satisfying (3) in Lemma~\ref{lem:peeling_root_chords} is given by
\[ 3 \int \hspace{-3pt} \int  
F_{3,\{ (2,1) \},\{4\} }(t) dt^2.\]
After recursive computations, and use of Lemma~\ref{lem:operator}, we find a formula for $F_{3,\{ (2,1),(2,1) \},\{4\} }(z)$, namely
\begin{multline*}
{\frac {1}{78274560}}\,{\upsilon}^{8}-{\frac {1}{3311616}}\,{\upsilon}^{7}
+{\frac {1}{168960}}\,{\upsilon}^{6}-{\frac {1}{55440}}\,{\upsilon}^{11/2}
+{\frac {1}{129024}}\,{\upsilon}^{5}+{\frac {1}{14112}}\,{\upsilon}^{9/2}-
{\frac {17}{86016}}\,{\upsilon}^{4} \\ +{\frac {1}{3920}}\,{\upsilon}^{7/2}-{
\frac {1}{6144}}\,{\upsilon}^{3}+{\frac {1}{10752}}\,{\upsilon}^{2}-{
\frac {1}{12012}}\,{\upsilon}^{3/2}+{\frac {169}{4515840}}\,\upsilon-{
\frac {1}{110880}}\,\sqrt {\upsilon}+{\frac {37}{39739392}}
\end{multline*}
where $\upsilon = 1-2z$.
\end{ex}

With this method, one can automatically compute every next-to$^k$ leading-log expansion for any $s$. For example, to compute it with $k=2$ and $s=2$, one should compute every generating function listed in Table~\ref{table:forntntlle} (plus the ones shown in Propositions~\ref{prop:lls2} and~\ref{prop:ntlle}).
One should find at the end the following expression.

\begin{table}[h!]
\begin{center}
\begin{tabular}{|c|c|} \hline
 Type & Generating function \\
 \hline

$(1,\varnothing,\{2,2\})$ & 
${\frac {17}{24}} - \left( \frac 1 8\, \ln  \left( 1-2\,z
 \right)  ^{2} + \frac 5 6 \right) \sqrt {1-2\,z}-\frac 1 2\,z
+\frac 1 6\, \left( 1-2\,z \right) ^{3/2}-\frac 1 {24}\, \left( 1-2\,z \right) ^{2}
$

 \\

$(1,\{(1,1)\},\{2\})$ & 
$
{\frac {41}{24}}- \left( \frac 1 4\,  \ln    \left( 1-2\,z
  \right) ^{2}+{\frac {11}{6}} \right) \sqrt {1-
2\,z}-\frac 3 2\,z+\frac 1 6 \, \left( 1-2\,z \right) ^{3/2}-\frac 1 {24}\, \left( 1-2\,z
 \right) ^{2}
$
 
 \\

$(1,\{(1,1),(1,1)\},\varnothing)$
&
$
1 - \left( \frac 1 8\, \ln \left( 1-2\,z 
  \right) ^{2}  + 1 \right) \sqrt {1-2\,z}-z
$

\\

$(2,\varnothing,\{2\})$ &
$
-\frac 3 8+ \left( \frac 1 3 + \frac 1 2\,\ln  \left(  1-2\,z 
 \right)  \right) \sqrt {1-2\,z}+\frac 3 2\,z+\frac 1 {24}\, \left( 1-2\,z \right) ^
{2}
$

\\

$(2,\{1,1\})$ &

$
-\frac 3 8+ \left( \frac 1 3 + \frac 1 2\,\ln  \left(  1-2\,z 
 \right)  \right) \sqrt {1-2\,z}+\frac 3 2\,z+\frac 1 {24}\, \left( 1-2\,z \right) ^
{2}
$

\\

 $(1,\{2,1\},\varnothing)$ &
$ 
{\frac {7}{8}}+ \frac 1 8\, \left( 1-2\,z \right) ^{2}-\frac 1 2\, \left( 1-2\,z
 \right) ^{3/2}-\frac 3 2\,z-\frac 1 2\,\sqrt {1-2\,z}

$

\\

 $(1,\{1,2\},\varnothing)$ &
$ 
{\frac {7}{8}}+ \frac 1 8\, \left( 1-2\,z \right) ^{2}-\frac 1 2\, \left( 1-2\,z
 \right) ^{3/2}-\frac 3 2\,z-\frac 1 2\,\sqrt {1-2\,z}

$

\\

 $(1,\varnothing,\{3\})$ & $
 
{\frac {7}{24}}+\frac 1 {24}\, \left( 1-2\,z \right) ^{2}-\frac 1 6\, \left( 1-2\,z
 \right) ^{3/2}-\frac 1 2\,z-\frac 1 6\,\sqrt {1-2\,z}

$
\\

$(3,\varnothing,\varnothing)$

&
$
\frac 3 8-\frac 1 {24}\, \left( 1-2\,z \right) ^{2}-\frac 1 2\,z-\frac 1 3\,\sqrt {1-2\,z}

$
\\
\hline

\end{tabular}
\end{center}

\caption{Generating functions involved in the next-to next-to leading log expansion ($s=2$)}
\label{table:forntntlle}
\end{table}

\begin{proposition}\label{NNLLs2}For $s=2$, the next-to next-to leading log expansion is given by
\begin{multline*} 
H_2(z) = {\frac { \left( a_{{2,0}}+a_{{1,1}}a_{{1,0}} \right) ^{2} \left( 
\ln  \left(  1+2\,za_{{1,0}}  \right)  \right) ^{2
}}{ 8 \left( 1+2\,za_{{1,0}} \right) ^{3/2}{a_{{1,0}}}^{2}}}-{
\frac { \left( a_{{2,0}}+a_{{1,1}}a_{{1,0}} \right) ^{2}\ln  \left( 
 1+2\,za_{{1,0}}  \right) }{2 \left( 1+2\,za_{{1,0}
} \right) ^{3/2}{a_{{1,0}}}^{2}}}\\-{\frac {z \left( -{a_{{2,0}}}^{2}+3
\,{a_{{1,0}}}^{3}a_{{1,2}}+3\,a_{{2,1}}{a_{{1,0}}}^{2}-a_{{1,1}}a_{{2,0
}}a_{{1,0}}+a_{{3,0}}a_{{1,0}} \right) }{ \left( 1+2\,za_{{1,0}}
 \right) ^{3/2}a_{{1,0}}}}.
\end{multline*}
\end{proposition}

Similarly, Table~\ref{table:s=1} lists every next-to$^k$ leading-log expansion for $s=1$ (QED photon self-energy, see the background for a discussion of this) and $k$ from $0$ to $3$.

\begin{table}[h!]
\begin{tabular}{|c|c|} \hline
 Expansion & Expression \\
 \hline

$H_0(z)$ & $-a_{1,0}z $ \\  \hline
$H_1(z)$ & \begin{minipage}{0.8 \textwidth}
\[-{\frac {\ln  \left( 1+za_{{1,0}} \right) a_{{2,0}}}{a_{{1,0}}}} \]
\end{minipage}
 \\ \hline
$H_2(z)$ & \begin{minipage}{0.8 \textwidth}
\[-{\frac {{a_{{2,0}}}^{2}\ln  \left( 1+za_{{1,0}} \right) }{{a_{{1,0}}}^
{2} \left( 1+za_{{1,0}} \right) }}-{\frac { \left( -{a_{{2,0}}}^{2}+a_
{{3,0}}a_{{1,0}}+a_{{2,1}}{a_{{1,0}}}^{2} \right) z}{a_{{1,0}} \left( 
1+za_{{1,0}} \right) }}
\] 
\end{minipage}

\\ \hline  $ H_3(z)$ &
\begin{minipage}{0.8 \textwidth}
\begin{multline*}
{-\frac {2\,{a_{{1,0}}}^{4}a_{{2,2}}+2\,{a_{{1,0}}}^{3}a_{{3,1}}-{
a_{{1,0}}}^{2}a_{{2,0}}a_{{2,1}}+{a_{{1,0}}}^{2}a_{{4,0}}-2\,a_{{1,0}}
a_{{2,0}}a_{{3,0}}+{a_{{2,0}}}^{3}}{2{a_{{1,0}}}^{3}}} \\ -{\frac {a_{{2,0}
} \left( {a_{{1,0}}}^{2}a_{{2,1}}+a_{{1,0}}a_{{3,0}}-{a_{{2,0}}}^{2}
 \right) }{{a_{{1,0}}}^{3} \left( za_{{1,0}}+1 \right) }}
+{\frac {{a_{{2,0}}}^{3} \left( \ln  \left( za_{{1,0}}+1 \right) 
 \right) ^{2}}{ 2 \left( za_{{1,0}}+1 \right) ^{2}{a_{{1,0}}}^{3}}} \\ + {\frac { \left( -2\,a_{{1,0}}a_{{2,0}}a_{{3,0}}-2\,{a_{{1,0}}}^{2}a_
{{2,0}}a_{{2,1}} \right) \ln  \left( za_{{1,0}}+1 \right) }{ 2 \left( za
_{{1,0}}+1 \right) ^{2}{a_{{1,0}}}^{3}}} \\ 
+{\frac {2\,{a_{{1,0}}}^{
4}a_{{2,2}}+2\,{a_{{1,0}}}^{3}a_{{3,1}}+{a_{{1,0}}}^{2}a_{{4,0}}+{a_{{
1,0}}}^{2}a_{{2,0}}a_{{2,1}}-{a_{{2,0}}}^{3}}{ 2 \left( za_{{1,0}}+1
 \right) ^{2}{a_{{1,0}}}^{3}}} 
\end{multline*}
\end{minipage} \\
\hline
\end{tabular}
\caption{Next-to$^k$ leading-log expansions $H_k(z)$ for $s=1$}
\label{table:s=1}
\end{table}

This generic method has been implemented for $s=1$ and $s=2$ in a \texttt{maple} file, which is available along the arXiv version of this paper.

\section{Dominance in the further leading-log expansions}\label{sec asymptotic}

\subsection{A dichotomy between $s=1$ and $s \geq 2$}

Remember Theorem~\ref{theo:intermsofwgf} where the next-to$^k$ leading-log expansion can be written as
\[ H_k(z) = \kappa - \hspace*{-15pt} \sum_{\substack{\textrm{all possible }(\Delta,\mathcal G,\mathcal D)\\\textrm{ such that } t(\Delta,\mathcal G,\mathcal D) \leq k }} A_{\Delta,\mathcal G,\mathcal D} {\dfrac {\partial^k F_{\Delta,\mathcal G,\mathcal D}}{\partial z^k}}\left( -a_{1,0} z \right),\]
where $A_{\Delta,\mathcal G,\mathcal D}$ is a monomial in the $a_{i,j}$'s. One natural question is to wonder which types of diagrams asymptotically contribute in the $n$th coefficient of $H_k(z)$ when $n$ grows to~$+\infty$. 

\begin{defi} A type $(\Delta,\mathcal G,\mathcal D)$ is \emph{dominant} in the next-to$^k$ leading-log expansion if \[ \lim_{n\rightarrow \infty} \left| \frac{  [z^n] {\dfrac {\partial^k F_{\Delta,\mathcal G,\mathcal D}}{\partial z^k}(z)}}{[z^n] H_k(z)}\right| = c > 0,\]
where $[z^n] F(z)$ denotes the coefficient of $z^n$ in the expansion of $F(z)$.
\end{defi}

Interestingly, the dominant types are not of the same form for $s=1$ and $s \geq 2$:

\begin{theo} Let $k$ be fixed.

For $s \geq 2$, the dominants types $(\Delta,\mathcal G,\mathcal D)$ in the next-to$^k$ leading-log expansion are of the form $(1,\{\underbrace{(1,1),\dots,(1,1)}_{k_1 \textrm{ times} }\},\{\underbrace{2,\dots,2}_{k_2 \textrm{ times} }\})$ with $k_1 + k_2 = k$.
In this case, 
\begin{equation}[z^n] { F_{\Delta,\mathcal G,\mathcal D}}\sim \frac {(s-1)^{k_1}} {\Gamma(1-\frac 1 s)  k!}  { k \choose k_1 } \ln(n)^k n^{-\frac 1 s-1} s^{n-k-1}. 
\label{eq:estimategeq2}
\end{equation}

For $s = 1$ and $k \geq 1$, the dominant type in the next-to$^k$ leading-log expansion is unique: it is $(2, \varnothing , \{\underbrace{2,\dots,2}_{k - 1 \textrm{ times} }\} )$.
We have
\begin{equation}
[z^n] {{F_{2, \varnothing , \{2,\dots,2\} }} } \sim \frac 1 { (k-1)! } \ln(n)^{k-1} n^{-2}. 
\label{eq:s1estimate}
\end{equation}
\label{theo:dominant_types}
\end{theo}

The proof of this theorem is sketched in the next subsection.

Remark that the estimate for $s=1$ can be seen as a limit $s \rightarrow 1$ of the estimate~\eqref{eq:estimategeq2} for $s=2$ (this is a purely formal remark since $s$ is an integer parameter and not a real one).

This theorem naturally gives an asymptotic estimate for the $n$th coefficient of the next-to$^k$ leading-log expansion, which extends the asymptotic results of \cite{CYchord}. 

\begin{cor} Let us assume that $a_{1,0} \neq 0$ and $ a_{2,0} + (s-1)a_{1,1}a_{1,0} \neq 0$. 

For $s \geq 2$, the $n$th coefficient of $H_k(z)$ grows asymptotically like
\begin{equation}
[z^n] H_k(z) \sim \frac {(-1)^{n}} {\Gamma(1-\frac 1 s)  k!}   a_{1,0}^{n-k} \left( a_{2,0} + (s-1)a_{1,1}a_{1,0} \right)^k  \ln(n)^k n^{k-\frac 1 s-1} s^{n-1},
\label{eq:hk_sgeq2}
\end{equation} 
while for $s=1$ and $k \geq 1$, the $n$th coefficient of $H_k(z)$ grows asymptotically like
\[ [z^n] H_k(z) \sim  \frac {(-1)^{n}} { (k-1)! }  (a_{1,0})^{n-k} a_{2,0}^k \ln(n)^{k-1} n^{k-2}. \] 
\end{cor}
\begin{proof} (from Theorem~\ref{theo:dominant_types}) Let us begin by $s \geq 2$. By keeping only the dominant types in Theorem~\ref{theo:intermsofwgf} (the other ones are negligible by definition), one get the asymptotic estimate
\[ [z^n] H_k(z) \sim \sum_{k_1 = 0}^k   {a_{1,0}}^{k-1}  a_{1,0} \left(\frac{a_{2, 0}}{a_{1,0}^2}\right)^{k_1}  \left(\frac {a_{1,1}}{a_{1,0}}\right)^{k-k_1}  
[z^n] {\dfrac {\partial^k F_{1, \{(1,1)\}^{k_1} ,\{2\}^{k-k_1}}}{\partial z^k}}\left( -a_{1,0} z \right)
,\]
where $\{x\}^{N}$ denotes the multiset containing the element $x$ $N$ times. Note that the $n$th  coefficient of something of the form  ${\dfrac {\partial^k f(z)}{\partial z^k}}$ is the $(n+k)$th coefficient of $f(z)$ multiplied by $(n+k)(n+k-1)\dots(n+1)$ (which is asymptotically equivalent to $n^k$).
Therefore the asymptotic estimate of Theorem~\ref{theo:dominant_types} leads (after simplification) to 
\[ [z^n] H_k(z) \sim \sum_{k_1 = 0}^k   {a_{1,0}}^{k} \left(\frac {a_{1,1}}{a_{1,0}}\right)^{k_1}   \left(\frac{a_{2, 0}}{a_{1,0}^2}\right)^{k - k_1}  
 \frac { (s-1)^{k_1}}{\Gamma(1-\frac 1 s)  k!}  { k \choose k_1 } \ln(n)^k n^{k-\frac 1 s-1} s^{n-1} (-a_{1,0})^n
.\]
But by the binomial theorem, $\sum_{k_1 = 0}^k   { k \choose k_1 }\left(\frac{a_{2, 0}}{a_{1,0}^2}\right)^{k_1}  (s-1)^{k_1} \left(\frac {a_{1,1}}{a_{1,0}}\right)^{k-k_1} = \left(\frac{a_{2,0}+(s-1)a_{1,1}a_{1,0}}{a_{1,0}^2}\right)^k$. One naturally recover estimate~\eqref{eq:hk_sgeq2}.

The proof is similar for $s=1$.
\end{proof}

Observe that only $a_{1,1},a_{2,0}$ and $a_{1,0}$ appear in the asymptotic expressions.  This means that asymptotically only these coefficients matter for any next-to${}^k$ leading-log expansion; that is only the residues of the one and two loop primitives and the next term in the expansion of the one loop primitive matter asymptotically.  This is in contrast to what one might expect of the full expansion of the Green function, as is discussed further in Section~\ref{subsec outlook}.

Additionally, the asymptotic forms are factored and the coefficients of the expansions of the primitives only appear via powers of $a_{1,0}$ and $(a_{2,0}+(s-1)s_{1,1}a_{1,0})$; no other expressions in these coefficients appear.  This is a highly simplified form compared to what might have been expected \textit{a priori}, though the physical significance of this form is not clear to us.

The hypotheses $a_{1,0}\neq 0$ and $a_{2,0} + (s-1)a_{1,1}a_{1,0}\neq 0$ are very reasonable physically.  The former says that there is a divergent 1-loop diagram, and the latter says that the different generating functions which can appear in the next-to-next-to leading log expansion for $s>1$ in fact do appear ---  there is no surprise simplification.  Should these hypotheses not hold in a situation of physical interest the leading contributions would be different, but a similar analysis to the above could be done to determine the asymptotic form.

It is tempting to consider summing the asymptotic forms in $k$ in order to get some sense of how $G(x,L)$ is behaving for large $n$.  However, even putting aside mathematical rigour, this is unlikely to give a good heuristic indication.  Not only is there no uniformity of the estimates in $k$, but in this situation we would expect, quite the opposite, that as $k$ grows the asymptotic form takes longer to become a good approximation (because of the presence of the logarithm).  As $k$ grows there are more diagram types in play and this would be expected to increasingly postpone the eventual dominance of the dominant types.

\subsection{Proof of Theorem~\ref{theo:dominant_types}} Since the proof of Theorem~\ref{theo:dominant_types} is very technical, only the main lines of the proof will be given here.

The main ingredient is here the transfer theorem from Flajolet and Odlyzko~\cite[Section VI.2.]{Flajolet-Sedgewick}. Provided some conditions of analyticity, this theorem states that the asymptotic behavior of the $n$th coefficient of a power series $f(z)$ results from the behaviour of $f(z)$ around its main singularity. In our context, the main singularity is simply the radius of convergence, which is $1/s$.

For example, if we observe that $f(z) - f(1/s) \sim_{z \rightarrow 1/s} (1-sz)^{-\alpha}$ (with $\alpha$ a complex number not in $\mathbb Z_{\leq 0}$), then the transfer theorem indicates that $[z^n] f(z) \sim \frac{n^{\alpha - 1}}{\Gamma(\alpha)} s^n$. Numerous asymptotic equivalences are given in ~\cite{Flajolet-Sedgewick} in this way, including the ones involving the logarithm. 

To use the transfer theorem (and thus prove Theorem~\ref{theo:dominant_types}), we proceed to an induction about the singular behaviour of the generating functions $F_{\Delta,\mathcal G,\mathcal D}$ for all types $(\Delta, \mathcal G, \mathcal D)$.

For $s \geq 2$, the induction establishes that if $(\Delta, \mathcal G, \mathcal D) = (1,\{\underbrace{(1,1),\dots,(1,1)}_{k_1 \textrm{ times} }\},\{\underbrace{2,\dots,2}_{k_2 \textrm{ times} }\})$, then
\begin{equation}
F_{\Delta,\mathcal G,\mathcal D}(z) - F_{\Delta,\mathcal G,\mathcal D}\left(1/s\right)  \sim_{z \rightarrow 1/s} - (s-1)^{k_1} \frac {s^{-(k_1+k_2+1)}} {(k_1 + k_2)!}  { k_1 + k_2 \choose k_1 } (1-sz)^{1/s} \ln(1-sz)^{k_1 + k_2}
\label{eq:sing2}
\end{equation} 
and if $(\Delta, \mathcal G, \mathcal D)$ is not of the form $(1,\{(1,1),\dots,(1,1)\},\{2,\dots,2\})$, then
\begin{equation}
F_{\Delta,\mathcal G,\mathcal D}(z) - F_{\Delta,\mathcal G,\mathcal D}\left(1/s\right) \sim_{z \rightarrow 1/s} C (1-sz)^{a + 1/s} \ln(1-sz)^{b},
\label{eq:22}
\end{equation} 
 where $a \geq 0$ and $ b < t(\Delta,\mathcal G,\mathcal D) =  \sum_{(d,g) \in \mathcal G} (d + g - 1) + \sum_{d \in \mathcal D} (d - 1) + \Delta - 1$ (cf Theorem~\ref{theo:intermsofwgf}), and $C$ a constant.

We use the generic method described in Subsection~\ref{ss:method} to derive the singular behavior of every function $F_{\Delta,\mathcal G,\mathcal D}(z) $. The base case of the induction is given by $\mathcal G = \mathcal D = \varnothing$. By some computations from items (1a) and (2) of the generic method, we get the estimate
\[ F_{\Delta,\mathcal G,\mathcal D}(z) - F_{\Delta,\mathcal G,\mathcal D}\left(1/s\right)  \sim_{z \rightarrow 1/s} -\frac{1}{\prod_{j=1}^\Delta js - 1} (1-zs)^{1/s},\]
which corroborates the induction hypothesis. Let us suppose now that $\mathcal D \neq \varnothing$ or $\mathcal G \neq \varnothing$. 
By using the induction hypotheses, we can proceed to some tedious analysis of items (1b) and (1c). Eventually one sees that if $(\Delta, \mathcal G, \mathcal D)$ is of the form $\left(1,\{(1,1),\dots,(1,1)\},\{2,\dots,2\}\right)$, then the singular behavior of $A'(z)$ is of the form $\kappa (1-zs)^{1/s}  \ln(1-sz)^{k_1+k_2-1}$, coming from the contributions of $F_{\Delta, \mathcal G, \mathcal D \backslash 2}$ (of singularity $\sim - \frac {(s-1)^{k_1} s^{-(k_1+k_2)}} {(k_1 + k_2 - 1)!}  { k_1 + k_2 - 1 \choose k_1 } (1-sz)^{1/s} \ln(1-sz)^{k_1 + k_2 - 1}$) 
and 
$F_{\Delta, \mathcal G \backslash (1,1), \mathcal D}$ (of singularity $\sim - (s-1) \frac {(s-1)^{k_1-1} s^{-(k_1+k_2)}} {(k_1 + k_2 - 1)!}  { k_1 + k_2 - 1 \choose k_1 - 1 } (1-sz)^{1/s} \ln(1-sz)^{k_1 + k_2 - 1}$).
By summing the two contributions, one can see that $\kappa = \frac {(s-1)^{k_1} s^{-(k_1+k_2)}} {(k_1 + k_2 - 1)!} \left ( { k_1 + k_2 - 1 \choose k_1 } + { k_1 + k_2 - 1 \choose k_1 - 1 } \right)$, which is $\kappa = \frac {(s-1)^{k_1} s^{-(k_1+k_2)}} {(k_1 + k_2 - 1)!} { k_1 + k_2 \choose k_1 } $ by Pascal's rule. By transformation (2) from the generic method, we deduce that
the singularity of $F_{\Delta, \mathcal G, \mathcal D}(z)$ is in 
\[ \kappa (1-sz)^{1/s} \int (1-st)^{-1} \ln(1-st)^{k_1 + k_2 + 1} dt  = \frac{\kappa}{s (k_1 + k_2)} (1-sz)^{1/s} \ln(1-sz)^{k_1 + k_2 }, \]
as predicted by \eqref{eq:sing2}.

As for types $(\Delta, \mathcal G, \mathcal D) \neq \left(1,\{(1,1),\dots,(1,1)\},\{2,\dots,2\}\right)$, we use a similar reasoning to show that the singular behavior \eqref{eq:22} of $F_{\Delta, \mathcal G, \mathcal D}$ is preserved. The induction is thus proved.

% 
%By the induction hypothesis, for any $d \in \mathcal D$, the main singularity of $F_{\Delta,\mathcal G,\mathcal D \backslash \{d\}}$ is of the form $C (1-sz)^{a} \ln(1-sz)^{b}$. After some calculations from item (1b) from the generic method, we find that the contribution  of $F_{\Delta,\mathcal G,\mathcal D \backslash \{d\}}$ to $A'(z)$ has for singular behaviour $C' (1-sz)^{a+d-2}  \ln(1-sz)^{b}$, where $C' = C$ if $d=2$.

To recover the statement of the theorem from the induction, we apply the transfer theorem to \eqref{eq:sing2}. It directly gives Estimate \eqref{eq:estimategeq2}. The analyticity of $F_{\Delta,\mathcal G,\mathcal D}(z) $ required by the transfer theorem is quite straightforward since the function is a polynomial in $(1-sz)^{1/s}$ and $\ln(1-sz)$ (this latter fact can also be proved by induction).

Finally, to show the dominance of the types $\left(1,\{(1,1),\dots,(1,1)\},\{2,\dots,2\}\right)$, we use again the transfer theorem to \eqref{eq:22}. We deduce that in the expression of $H_k$ (Theorem~\ref{theo:intermsofwgf}), only the terms in $ \frac {\partial^k F_{1,\{(1,1),\dots,(1,1)\},\{2,\dots,2\}}^k } {\partial z^k}$ prevail asymptotically.

Concerning $s=1$, a similar induction holds. It states that for  $(\Delta, \mathcal G, \mathcal D) = (2,\varnothing,\{\underbrace{2,\dots,2}_{k-1 \textrm{ times} }\})$, 
\[ F_{\Delta,\mathcal G,\mathcal D}(z) - F_{\Delta,\mathcal G,\mathcal D}\left(1\right)  \sim_{z \rightarrow 1} - \frac 1 {(k-1)!} (1-z) \ln(1-z)^k \]
and for other type $(\Delta, \mathcal G, \mathcal D)$,
\[ F_{\Delta,\mathcal G,\mathcal D}(z) - F_{\Delta,\mathcal G,\mathcal D}\left(1\right)   =_{z \rightarrow 1} o \left( (1-z)\ln(1-z)^{t(\Delta,\mathcal G,\mathcal D)} \right).\]
Why are there these differences from $s \geq 2$? When counting diagrams satisfying (1c) in the generic method, we multiply our coefficients by $sm - 1$. But if $m=1$, which occurs whenever  $(\delta,\gamma) = (1,1)$, then we multiply by $0$. That is why the main contribution for $s=1$ to the singularity of $H_k$ cannot come from a diagram of type $(\Delta,\mathcal G,\mathcal D)$ with $(1,1) \in \mathcal G$. The same method applies otherwise; one similarly proves the induction and then the theorem.

\section{Discussion}

\subsection{Comparison with Kr\"uger and Kreimer}\label{subsec compare}

In \cite{KKllog} Kr\"uger and Kreimer developed an approach to calculating
next-to${}^{j}$-leading log expansions, using a modified shuffle structure on the new numbers appearing.  As with our techniques, they break each log expansion into subexpansions, these subexpansions can be interpreted as generating functions for certain combinatorial objects, and appropriately weighting these generating functions and summing gives the result.  The combinatorial objects in their case are words.  Roughly, the words are over the alphabet of primitive Feynman diagrams in the theory under study, but they also must augment their alphabet with new letters corresponding to commutators of old letters as well as new letters coming from the quasi-shuffle terms of their shuffles.

Their approach is automatable with two caveats.  First, the generating functions involving these new commutator letters do not fit into their master theorem and so do not have a fully automatic theory worked out for them; this is the case of indexed matrices in their language.  Second, their master theorem involves some coefficients which come from counting certain matrices, but a general enumerative theory for this has not been worked out.

Their techniques, as ours do, apply to quantum field theories quite generally and they worked out the details for the propagator in Yukawa theory and the photon self-energy in quantum electrodynamics up to the next-to-next-to-leading log order.  These examples correspond to $s=2$ and $s=1$, respectively, in our set up.  Their techniques also have the advantage of being tightly related to the Lie algebra structure of Feynman diagrams.

\medskip

For the purposes of comparing our work with theirs, it is worth going into a bit more detail on the form of their results.  They express each log expansion as a sum of products of a \emph{period} and a generating function.  Their generating functions, like our types, are based on having any number of the smallest contributions (in their case shuffles with their base letter $a$, in our case nonterminal chord of decoration 1). Their generating functions are, in our language, functions of the single variable $xLa_{1,0}$.  Note that all of our generating functions depend on $x$ and $L$ only via the product $xLa_{1,0}$ and what remains is a factor which is a small rational function of the $a_{i,j}$.  In this way we also have a product of a generating function in $xLa_{1,0}$ and another factor.  For them, this other factor is a period, and they identify each period (non-canonically) with a linear combination of products of Feynman graphs where the period then is the Feynman period of the corresponding combination of Feynman graphs.  Feynman periods \cite{Brbig,Sphi4} are certain renormalization scheme independent residues of Feynman integrals; they are periods in the sense of Kontsevich and Zagier \cite{Ko-Za}, that is they are integrals of algebraic functions over algebraic domains.

In particular, in our notation, the primitive (or sum of primitives) at $k$ loops has $F_k(\rho) = a_{k,0}/\rho + a_{k,1} + a_{k,2}\rho + a_{k,3}\rho^2 + \cdots$ as the expansion of its regularized Feynman integral, and so the period of this primitive is $a_{k,0}$.  These periods correspond to the letters which Kr\"uger and Kreimer notate $a_k$ in their words.

As noted above, they also have additional letters in their words, and correspondingly additional periods, from commutators and quasi-shuffle terms.  These no longer correspond to individual $a_{k,i}$ for us, but rather to combinations.  For example the quasi-shuffle term coming from shuffling two copies of the one-loop primitive is for them notated $\Theta(a_1,a_1)$.  It corresponds to the linear combination of graphs
\[
\includegraphics{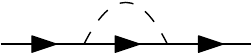}\ \includegraphics{images/Yukawad} - 2\includegraphics{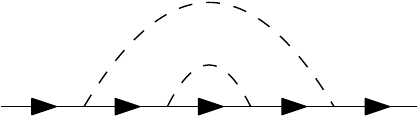}
\]
in Yukawa theory, and similarly in other theories.  The regularized Feynman integrals of each term in this sum have $1/\rho^2$ terms, but these cancel, leaving only a simple pole with residue $a_{1,1}$ in our notation.  In more complicated examples, these additional periods of Kr\"uger and Kreimer will correspond to linear combinations of products of our $a_{i,j}$.

\medskip

To compare our results to theirs, we need to consider how their periods line up with the different combinations of the $a_{i,j}$ and to consider how the remaining generating functions in $xLa_{1,0}$ line up.  In both cases there is a discrepancy. 
However, through personal discussion with Kr\"uger and Kreimer (personal communication, 2018) we believe that we understand this discrepancy.

Our results are as found in Table~\ref{table:s=1} and Propositions~\ref{prop:lls2}, \ref{prop:ntlle}, and \ref{NNLLs2}.  See Section A.4 of \cite{KKllog} for an analogous statement of their results.
We see a precise correspondence in the leading log and next-to leading log results.

At next-to next-to leading order we begin to see differences.  First, no matter how we try to align expressions in our $a_{i,j}$ with their periods we cannot get the same generating functions as we do not see terms with $(1-2xLa_{1,0})^{-5/2}$ in the Yukawa case, nor $(1-xLa_{1,0})^{-2}$ in the QED case.  For Kr\"uger and Kreimer these terms came from generating series for words where one of the letters is a commutator of other letters.  From personal discussion with Kr\"uger, this is due to an error in their paper.  Their general methods are correct, but in working out their word reduction algorithm certain systematic errors were made that impacted these cases.

The second important difference is the number of periods.  They have eight periods in the Yukawa case (fewer in QED as some things vanish).  These should correspond to some sums of products of $a_{i,j}$s.  However, we see only $a_{3,0}$, $a_{2,0}^2$, $a_{2,1}$, $a_{1,2}$, $a_{1,1}a_{2,0}$, $a_{1,1}^2$, all with appropriate powers of $a_{1,0}$ as well, and this is all that can show up for reasons of homogeneity.  In general, Kr\"uger and Kreimer will see as many periods as rooted trees, but we will only see as many expressions in the $a_{i,j}$ as appropriately weighted partitions, so we simply cannot find as many periods as they do asymptotically.  Having discussed with Kreimer and Kr\"uger, we believe that their results make no assumptions on what the Feynman rules are, and so every different diagram, every commutator, and every quasi-shuffle term, being algebraically different is treated as an independent period.  Our approach, by contrast, builds all the Feynman integrals through the Dyson-Schwinger equation out of the $F_k(\rho)$. Consequently we have built right into our setup additional information which tells us how each Kr\"uger-Kreimer period can be written in terms of the $a_{i,j}$.  These expressions for the Kr\"uger-Kreimer periods are not algebraically independent, but satisfy algebraic relations.

\medskip

Subsequently, in the process of correcting these errors Kr\"uger has found an alternate approach which comes more directly from the renormalization group equation and is the subject of work in preparation \cite{Krge}. %I'd like to put this cite here, but I need a title from Olaf: \cite{Krge}.

\subsection{Gauge theory dichotomy}

Running throughout the results of this paper, there has been a dichotomy between the $s=1$ case and the $s>1$ case.

The $s=1$ case corresponds to propagator corrections for a particular such as the photon in QED, because $s=1$ is the case where the one-loop correction has no internal line of its own type, and one additional such internal line appears at each subsequent loop order.

The gluon in QCD also has an $s=1$ type of behaviour, though for more complicated reasons. The one-loop gluon correction with a fermion (or ghost) bubble is straightforwardly $s=1$.  The one-loop gluon correction with a gluon bubble is in and of itself an $s=3$ case (as each additional gluon edge with further 3-valent vertices would add three more gluon insertion places; the same combinatorics as in a scalar $\phi^3$ theory.)  However, the different vertices in QCD do not have independent coupling constants.  In the Hopf algebraic approach to Feynman diagrams, this is dealt with by taking the combinatorial versions of the invariant charges for each vertex and then working in the quotient Hopf algebra given by modding out by the ideal generated by identifying these invariant charges.  The combinatorial versions of the invariant charges are discussed in section 3.3.2 of \cite{kythesis}; the mathematics and physics of these quotient algebras are discussed in the context of Ward and Slavnov-Taylor identities in \cite{anatomy} and \cite{vS}.  In order to speak of a single value of $s$ for the gluon, we need to work in this quotient algebra\footnote{An analysis not working in a quotient algebra should also be possible but would involve machinery similar to that which will be needed for extending the chord diagram expansions to systems of Dyson-Schwinger equations.  In both cases this remains work for the future.}, and in this algebra since the gluon bubble can be expressed in terms of other one-loop gluon corrections, we end up with an $s=1$ behaviour as we did for the QED photon.  

Consequently, the $s=1$ vs $s>1$ dichotomy can be viewed as a gauge particle vs non-gauge particle dichotomy.  In particular, in the $s=1$ case we see substantial simplifications.  The form of the generating functions for the $N^kLL$ expansions is strictly simpler in $s=1$ compared to $s>1$.  This can be seen in the explicit results we give up to the next-to-next-to order where there are simply fewer terms in the closed forms of the functions for $s=1$ when compared to $s>1$.  We also see simplifications in the asymptotic forms.  Fewer classes of chord diagrams are asymptotically relevant for $s=1$ and correspondingly fewer of the $a_{i,j}$ are involved in the dominant contributions.  Specifically only $a_{1,0}$ and $a_{2,0}$ appear in the asymptotic expression for $s=1$ while $a_{1,0}$, $a_{2,0}$ and $a_{1,1}$ appear in the asymptotic expression for $s>1$.

This is consistent with the fact that physicists have long observed extra cancellations in gauge theory expansions, see for instance the summary given by Le Guillou and Zinn-Justin in  \cite[ Sections 2 and 3]{GZJlargeorder} as well as the papers reprinted later in that volume.

Note that $s=1$ refers specifically to the gauge particle, while other particles in a gauge theory will typically have $s>1$, e.g. the QED fermion has $s=2$.  Thus, the dichotomy is between gauge particles and other particles rather than between gauge theories and non-gauge theories. 

Note also that we are not directly using the gauge invariance, and so particles with the same Feynman diagram combinatorics as a gauge particle would also fall into the $s=1$ case.  Nor can the $s$ parameter distinguish a toy like scalar electrodynamics from a realistic gauge theory.  These distinctions will appear only in the $a_{i,j}$'s.  Further effects of the gauge group on the asymp\-totic picture would necessarily come from additional cancellations or identities between the $a_{i,j}$.  It is particularly striking, however, that we are already seeing simplifications and a clear dichotomy arising out of the diagram combinatorics for gauge particles. We hope this is a useful framework to clarify the understanding of gauge theory asymptotics.

\subsection{Outlook}\label{subsec outlook}

All of our results are systematic and automatable giving an algorithmic procedure for calculating any $N^kLL$ expansion as well as its asymptotics.  \texttt{Maple} code is included with the auxiliary files in the arXiv version of this paper.

The chord diagram expansion is an expansion for the Green function itself, not just for the various $N^kLL$ expansions.  Consequently, in principle we can access information on the large order behavior of the full Green function using this expansion.
The combinatorial analysis is more difficult for the Green function as a whole, principally because we lose the restriction on the location of the first terminal chord, and this restriction was very useful in all of our arguments.  

In Section~\ref{sec asymptotic} we analyzed that asymptotic behavior of the $N^kLL$ expansions and showed that only the $a_{1,1}, a_{2,0}$, and $a_{1,0}$ appear in the asymptotic expression.

Such results, however, are not sufficient to understand the asymptotic behavior of the coefficients in the expansion of the Green function unless we also understand how the $a_{i,j}$ grow.  This consideration is not important for the $N^kLL$ expansions because the $N^kLL$ expansion can involve only $a_{i,j}$ for $i,j\leq k$, so there is an automatic truncation simply by the nature of the expansion.   For the full Green function, however, if $a_{k,0}$ or $a_{1,k}$, or other subsequences grow fast enough, then terms involving these elements could dominate even with other terms dominating each $N^kLL$ expansion.

Our methods will not be able to shed light on the growth of the $a_{i,j}$ as this is a more strictly analytic concern.  If we assume all the $a_{i,j}$ are roughly of the same order then from the $N^kLL$ results the $a_{1,0}$, $a_{1,1}$, $a_{2,0}$ would also dominate in the full Green function.  Rigorous results in this direction are likely accessible, though we did not pursue this here.

Perhaps more interesting is the potential for conditional results which would say, for instance, that if the terms involving $a_{k,0}$ with $k\rightarrow \infty$ are asymptotically important, then the sequence $a_{k,0}$ must grow at at least some specified rate while the other coefficients remain bounded.  Results of this form for the $a_{k,0}$ or other subsequences would be interesting because the different subsequences correspond to different sources of divergence.  The dominance of the $a_{k,0}$ would correspond to asymptotic dominance by primitive Feynman diagrams as the $a_{k,0}$ are the leading terms for each primitive.  What is currently known about whether or not the contributions of the primitive diagrams dominates is explained in section 4 of \cite{MKlargeorder}; while it remains conjectural, McKane therein explains that in the 80s the contribution of the primitives was generally expected to dominate, and that this remains plausible but unproven from today's perspective.  Dominance of the low terms, $a_{1,0}$, $a_{1,1}$, $a_{2,0}$ etc., corresponds to renormalon-type behavior as these terms come from small Feynman diagrams iterated into themselves repeatedly.  It is less clear what dominance by the $a_{1,k}$ or other subsequences would mean.  Perhaps, via resurgence, some of these subsequences show the instanton behavior of the theory.  The conditions in such conditional results could then be compared to asymptotic results or estimates obtained in other ways.

We can also ask what bearing the present results have on past work of one of us with other coauthors \cite{vBKUY, vBKUY2}.  There we worked with, and required assumptions on, the mysterious function $P(x)$ first defined in \cite{kythesis}.  If we write $G(x,L) = 1-\sum_{i\geq 1} \gamma_i(x) L^i$ then
\[
P(x) = \gamma_1 + 2\gamma_2.
\]
  {}From the chord diagram expansion, for $G(x,L)$ satisfying a Dyson-Schwinger equation of the form \eqref{eq:gen case}, we have
  \[
  P(x) = \sum_{\substack{\omega_s\text{-marked}\\\text{diagram } C}} (a_{d(t_1), t_1-1}
  - a_{d(t_1), t_1-2})A(C) x^{\|C\|} 
  \]
  where the $t_1$ appearing in the sum is $t_1(C)$.  This equation shows that $P(x)$ is closely related to, but simpler than, the full Green function, as we sum over the same diagrams as for the full Green function, but in place of the polynomial in $L$ we have only a difference of adjacent $F_{d(t_1)}(\rho)$ coefficients, and this difference depends only on the index of the first terminal chord.  This expression for $P(x)$ is not as easy to work with as the N${}^k$LL expansions since we cannot restrict to diagrams with large first terminal chord. However it has a decomposition based on the first terminal chord, making it appear more approachable than the full Green function.  For \cite{vBKUY} and \cite{vBKUY2} we required assumptions on the growth of $P(x)$ as a function of $x$ for large $x$.  Chord diagram techniques only tell us about the series expansion of $P(x)$ around $0$, so we would need some resummation results to move from asymptotic results of the latter type to those of the former type.  However, we may be lucky if the differences $a_{d(t_1), t_1-1} - a_{d(t_1), t_1-2}$ are often 0 or are small and approaching 0 sufficiently quickly in the relevant limits.  Then the convergence of the series for $P(x)$ may be better than it naively appears.  This also remains a question for the future.

\bibliographystyle{plain}
\bibliography{expansions}

\end{document}